\documentclass[10pt, onecolumn]{IEEEtran}

\IEEEoverridecommandlockouts
\usepackage{amsmath}
\usepackage{mathrsfs}
\usepackage{amsfonts}
\usepackage{graphicx}
\usepackage{amssymb}
\usepackage{latexsym}
\usepackage{array}
\usepackage{ifthen}
\usepackage{wrapfig}
\usepackage{picinpar}
\usepackage{epic}
\usepackage{rotating}
\usepackage{morefloats}
\usepackage{enumerate}
\usepackage{bm}
\usepackage{stfloats}
\usepackage{dsfont}
\usepackage[english]{babel}
\usepackage{cite}
\usepackage{color}
\usepackage{enumerate}
\usepackage{subfig}
\newtheorem{theorem}{Theorem}
\newtheorem{lemma}{Lemma}

\newtheorem{corollary}{Corollary}
\newtheorem{definition}{Definition}
\newtheorem{remark}{Remark}
\newtheorem{example}{Example}

\def\bfc{\mbox{\boldmath$c$}}
\def\bfv{\mbox{\boldmath$v$}}
\def\bfi{\mbox{\boldmath$i$}}

\def\bfu{\mbox{\boldmath$u$}}

\def\qed{ \rule{.08in}{.08in}}

\title{\LARGE \bf {On} Spectral Properties of Signed Laplacians with\\{Connections to Eventual Positivity}\thanks{This work was supported
    in parts by the Research Grants Council of Hong Kong Special Administrative Region, China, under the Theme-Based Research Scheme T23-701/14-N, the Knut and Alice Wallenberg Foundation, the Swedish Research Council, the National Science Foundation under Grant 1901599, and a MURI grant at the University of Illinois.}}

\author{Wei~Chen,~Dan~Wang,~Ji~Liu,~Yongxin~Chen,~Sei~Zhen~Khong,~Tamer~Ba\c{s}ar,\\~Karl~H.~Johansson,~and~Li~Qiu\vspace{-5pt}
\thanks{W. Chen is with the Department of Mechanics and Engineering Science \& Beijing Innovation Center for Engineering Science and Advanced Technology, Peking University, Beijing 100871, China. {\tt w.chen@pku.edu.cn}}
\thanks{D. Wang and L. Qiu are with the Department of Electronic and Computer Engineering, Hong Kong University of Science and Technology, Clear Water Bay, Kowloon, Hong Kong, China. {\tt dwangah@connect.ust.hk, eeqiu@ust.hk}}
\thanks{J. Liu is with the Department of Electrical and Computer Engineering, Stony Brook University, Stony Brook, NY 11794-2350, USA. {\tt ji.liu@stonybrook.edu}}
\thanks{Y. Chen is with the School of Aerospace Engineering, Georgia Institute of Technology, Atlanta, GA 30332, USA. {\tt yongchen@gatech.edu}}
\thanks{S. Z. Khong is an independent researcher. {\tt szkhongwork@gmail.com}}
\thanks{T. Ba\c{s}ar is with the Coordinated Science Laboratory, University of Illinois at
    Urbana-Champaign, Urbana, IL 61801, USA.  {\tt basar1@illinois.edu}}
\thanks{K. H. Johansson is with the School of Electrical Engineering and Computer Science, KTH Royal
Institute of Technology, Stockholm, Sweden. {\tt kallej@kth.se}}
    }

\begin{document}

\maketitle
\thispagestyle{empty}
\pagestyle{empty}


\begin{abstract}
Signed graphs have appeared in a broad variety of applications, ranging from social networks to biological networks, from distributed control and computation to power systems. In this paper, we investigate spectral properties of signed Laplacians for undirected signed graphs. We find conditions on the negative weights under which a signed Laplacian is positive semidefinite via the Kron reduction and multiport network theory. For signed Laplacians that are indefinite, we characterize their inertias with the same framework. Furthermore, we build connections between signed Laplacians, generalized M-matrices, and eventually exponentially positive matrices.
\end{abstract}

\section{Introduction}
A signed weighted graph is a graph whose nodes are linked by edges of positive and negative weights. Research of signed graphs can be traced back to Fritz Heider's psychological study on interpersonal relations, where positive and negative weights represent liking and disliking among individuals \cite{heider1946attitudes}. This stimulated the interest of mathematician Frank Harary who introduced the notion of balance of a signed graph in 1953 \cite{harary1953notion}. These pioneering works have led to psychological and sociological studies by means of the mathematical tool of signed graphs. We refer interested readers to \cite{taylor1970balance} and references therein for the advances until the late 1960s.

Recently, research on signed weighted graphs has seen a revival, driven by applications in a broad range of areas including opinion dynamics {\cite{altafini,altafini2015predictable,shi2016evolution,proskurnikov2016opinion,liu2017expo,liu2018polarizability,SAB}}, distributed control and optimization \cite{Boyd2004,Zelazo,yxchen,zelazo2017,pan2016laplacian,ahmadizadeh2017eigenvalues,zhang2017bipartite}, data clustering and graph-based machine learning \cite{kunegis2010spectral}, biological networks \cite{nishikawa2010network}, power systems {\cite{motter2013spontaneous,song2017network,ding2017impact}}, and knot theory \cite{lien2000dual}. See also \cite{SAB} for a recent review on convergence properties of dynamics over deterministic or random signed networks.

To be specific, negative edge weights have been employed to represent antagonistic relations in social networks, anticorrelation among data in clustering, and inhibitory interactions in interneuron networks. In the study of small disturbance angle stability of power systems, {negative edge weights may occur due to some transmission lines with negative reactance or large angle differences across certain transmission lines \cite{motter2013spontaneous,song2017network}.}
In distributed control and optimization, negative weights may stem from faulty communication among agents or adversarial attacks on the network. In some cases, negative weights even arise as a result of optimal design. References \cite{Boyd2004} and \cite{nishikawa2010network} have shown that allowing negative weights in the design may have positive effects on accelerating the convergence in both distributed averaging and synchronization.

All in all, there is abundant motivation to study signed weighted graphs.
It is often the case that studying dynamics over a signed weighted graph requires the analysis of spectral properties of an associated signed Laplacian matrix.
Consider, for example, the case of a continuous-time multi-agent system interconnected over a signed graph aiming to reach consensus. Under certain protocol, consensus can be reached if and only if the associated signed Laplacian has all its eigenvalues in the open right half complex plane except for a simple zero eigenvalue.
In fact, many existing consensus-based distributed control, estimation and optimization algorithms, which are initially designed for conventional weighted graphs can be extended to signed weighted graphs by replacing conventional Laplacians with signed ones. A necessary condition for those algorithms still to function correctly over a signed weighted graph is that the corresponding signed Laplacian has a simple zero eigenvalue and all the other eigenvalues have positive real parts.
Signed Laplacians have attracted increasing attention recently.
In this paper, we focus on undirected signed weighted graphs for which the preceding spectral condition on signed Laplacians simplifies to positive semidefiniteness with a simple zero eigenvalue.
Below, we briefly review some closely related works on undirected signed Laplacians.

A fundamental issue frequently discussed in the literature is the positive semidefiniteness of signed Laplacians. Exploring conditions rendering signed Laplacians positive semidefinite is of great importance in many applications.
{It was shown in \cite{Zelazo} that a signed Laplacian with a single negative edge weight is positive semidefinite if and only if the absolute value of the negative edge weight is less than or equal to the reciprocal of the effective resistance between the nodes of the negative edge over the positive subgraph. Therein the authors extended this condition to signed graphs with multiple negative edges under certain additional constraints on the locations of the negative edges.} These results were re-established in \cite{yxchen} using both a geometrical approach and a passivity-based approach. { In \cite{zelazo2017}, the more general case of a signed graph with multiple negative edges and no restrictions on the locations of negative edges was considered, where necessary and sufficient conditions on the semidefiniteness of signed Laplacians were obtained via linear matrix inequalities (LMIs). Similar LMI conditions also appeared in \cite{song2017network,ding2017impact}. However, in the most general case, an explicit condition given in terms of effective resistances has not been available in these papers.}

When a signed Laplacian matrix is not positive semidefinite, its inertia, i.e., the numbers of its negative, zero, and positive eigenvalues with multiplicity counted, often plays an important role in applications. It is known that the type of an unstable equilibrium point in a power system is decided by the inertia of a certain signed Laplacian \cite{song2017network}. {Reference \cite{la1} obtained bounds on the inertia of a signed Laplacian based on the topology of the graph.} Reference \cite{pan2016laplacian} considered the problem of how the structure of a signed graph influences the inertia of the associated signed Laplacian.

Another issue worthy of attention stems from the observation that a signed Laplacian is not an M-matrix\footnote{A square matrix $M$ is said to be an M-matrix if it can be expressed as $M=sI-B$, where $I$ is the identity matrix, $B$ is nonnegative, and $s\ge \rho(B)$.} as opposed to the conventional Laplacian with only positive weights. Many nice properties of M-matrices are inherited from nonnegative matrices. It is well known that a nonnegative matrix possesses the Perron-Frobenius property \cite{HorJoh85}, i.e., its spectral radius is an eigenvalue with a corresponding nonnegative eigenvector. However, a matrix possessing the Perron-Frobenius property is not necessarily nonnegative. For this reason, much interest has been generated in exploring so-called eventually nonnegative matrices, namely, matrices which become nonnegative after a certain finite power \cite{noutsos2008reachability,noutsoslaa}. Furthermore, based on eventual nonnegativity, various generalized M-matrices have been proposed and studied \cite{elhashash2008generalizations,olesky2009m}. {Recently, the interplay between eventual nonnegativity and multi-agent consensus has also been considered \cite{altafini2015predictable,jiang2016sign}.} In light of these developments, it is of interest to study whether a signed Laplacian belongs to a class of generalized M-matrices.

{In this paper, we present results in the three topical areas identified above. The main contributions of the paper can be summarized as follows:

\begin{enumerate}[(a)]
\item We study semidefiniteness of signed Laplacians with multiple negative weights and no restrictions on where the negative edges are located with the $n$-port network theory, yielding explicit necessary and sufficient conditions in terms of conductance (or resistance) matrices. The inertia of an indefinite signed Laplacian is also characterized via the conductance matrix.
\item We further characterize semidefiniteness and inertia of signed Laplacians via the Kron reduction, a seminal tool in power systems.
\item We establish connections between signed Laplacians, generalized M-matrices, and eventual positivity.
\end{enumerate}
Part of the results regarding (a) have appeared in the authors' conference papers \cite{ChenCDC16,chen2016semidefiniteness,chen2017spectral}. These results are now unified under the framework of $n$-port network. Some proofs not included in the conference versions are now included here and the unnecessary assumption adopted in \cite{chen2017spectral} when characterizing the inertia is also relaxed. The contributions (b) and (c) are new and go much beyond the scope of discussions in the conference papers. The new contributions provide a more comprehensive view of the spectral properties of signed Laplacians and their connections to eventual positivity.}

The rest of the paper is organized as follows. {The problem setup and motivating applications are introduced in Section II.} Some preliminaries are given in Section III. Positive semidefiniteness of signed Laplacians is investigated in Section IV, followed by characterization of the inertias of indefinite signed Laplacians in Section V. Connections between signed Laplacians and eventual positivity are discussed in Section VI. The paper is concluded in Section VII.

{\em Notation:} Denote by $\mathbf{1}$ the vector with all its
elements equal to $1$, where the dimension is to be understood from the context. Denote by $d_i$ the vector whose $i$th element is $1$ and all the other elements are $0$. Also, let $d_{ij} \!=\! d_i-d_j$. Denote the spectral radius of a square matrix $A$ by $\rho(A)$. {Denote the corank of a matrix $A\in\mathbb{R}^{n\times n}$ by $\mathrm{corank}(A)=n-\mathrm{rank}(A)$.} A square matrix $A$ is said to be nonnegative (positive, respectively), denoted by $A\unrhd0$ ($A\rhd 0$, respectively), if all the elements of $A$ are nonnegative (positive, respectively). {A square matrix $A$ is said to be exponentially nonnegative (exponentially positive, respectively) if $e^{At}=\sum_{k=0}^{\infty}\frac{(tA)^k}{k!}\unrhd 0$ ($e^{At}\rhd 0$, respectively) for all $t> 0$.} Denote by $|\alpha|$ the cardinality of a set $\alpha$ and by $\alpha\backslash\beta$ \mbox{the relative complement of a set $\beta$ in $\alpha$}.

For a real symmetric matrix $S$, we write $S\ge 0$  when $S$ is
positive semidefinite, and $S>0$ when $S$ is positive definite. Denote by $\pi(S)=\{\pi_-(S),\pi_0(S),\pi_+(S)\}$ the inertia of a real symmetric matrix $S$, where $\pi_-(S)$, $\pi_0(S)$, and $\pi_+(S)$ are respectively the numbers of negative, zero, and positive eigenvalues with multiplicity counted. For a matrix
\[S=\begin{bmatrix}S_{11}&S_{12}\\S_{21}&S_{22}\end{bmatrix},\]
denote by $S\slash_{22}=S_{11}-S_{12}S^{\dagger}_{22}S_{21}$ the (generalized) Schur complement of $S_{22}$ in $S$, where $S^{\dagger}_{22}$ means the Moore-Penrose pseudoinverse of $S_{22}$. Similarly, $S\slash_{11}=S_{22}-S_{21}S^{\dagger}_{11}S_{12}$.

\section{{Problem Setup and Applications}} \label{wl}
\subsection{Signed graphs and signed Laplacians}
Consider an undirected graph $\mathbb{G}\!=\!(\mathcal{V},\mathcal{E})$ with a set of nodes $\mathcal{V}\!=\!\{1,2,\dots,n\}$ and a set of edges
$\mathcal{E}\!=\!\{e_1,e_2,\dots,e_m\}$.
We use $(i,j)$ to represent the edge connecting node $i$ and node $j$, and associate with each edge $(i,j)$ a real-valued nonzero weight $a_{ij}$, which can be either positive or negative. If node $i$ and node $j$ are not connected by an edge, $a_{ij}$ is understood to be zero.
Such a graph is called a {\em signed weighted graph}. For brevity, hereinafter signed weighted graphs are referred to simply as signed graphs.

For an undirected graph $\mathbb{G}$, a spanning tree, i.e., a spanning subgraph\footnote{A spanning subgraph of $\mathbb{G}$ is a graph which contains the same set of nodes as $\mathbb{G}$ and whose edge set is a subset of that of $\mathbb{G}$.} which itself is a tree, exists if and only if $\mathbb{G}$ is connected. A spanning forest is a spanning subgraph containing a spanning tree in each connected component of the graph. A spanning tree can be regarded as a special case of a spanning forest. Hereinafter we use
$\mathbb{F}$ to represent a spanning tree or a spanning forest, depending on whether the underlying graph is connected or not.

For a signed graph, the associated {\em signed Laplacian matrix}
 $L\!=\![l_{ij}]\!\in\!\mathbb{R}^{n\times n}$ is defined by \cite{altafini2015predictable,Zelazo,la1}
\begin{align*}
l_{ij}=\begin{cases}-a_{ij},&i\neq j,\\ \sum_{j=1,j\neq i}^n a_{ij},& i=j.\end{cases}
\end{align*}
A signed Laplacian matrix is no different from a conventional one, except that the conventional Laplacian has only positive weights. Some properties known for conventional Laplacian matrices continue to hold in the presence of negative weights.
For instance, a signed Laplacian $L$ is symmetric, and hence all the eigenvalues are real. Also, $L$ has a zero eigenvalue with a corresponding eigenvector being $\mathbf{1}\in\mathbb{R}^n$.

However, signed Laplacians also carry some fundamental differences. First, unlike the conventional Laplacians that are positive semidefinite, signed Laplacians may be indefinite. Second, while the multiplicity of zero eigenvalue of a conventional
Laplacian equals the number of connected components in the underlying graph, this is in general not true for a signed Laplacian. Third, any conventional Laplacian is an M-matrix and its negation is exponentially nonnegative \cite{berman1994nonnegative}. This is no longer the case for a signed Laplacian.

Considering the aforementioned similarities and differences, we address in this paper the following questions regarding the signed Laplacian $L$:
\begin{enumerate}[1)]
\item How to characterize the negative weights under which $L$ is positive semidefinite with a simple zero eigenvalue?
\item In case $L$ is indefinite, is there a simple way to characterize its inertia?
\item Is it possible that under some conditions, $L$ is some generalized M-matrix and $-L$ is some generalized exponentially nonnegative matrix? What kind of generalized M-matrices and exponentially nonnegative matrices can we consider?
\end{enumerate}

{These questions are not only of interest from a mathematical perspective, but also of importance in many applications. To further motivate this study, we shall discuss some representative applications in the following subsections.}

Before proceeding, we introduce a useful factorization of $L$. Let
$W\!=\!\mathrm{diag}\{w_1,w_2,\dots,w_m\}$, where $w_{k}\!=\!a_{ij}, \text{ for }(i,j)=e_k$.
Also, assign an (arbitrary) orientation to each edge $e_k$ by denoting one endpoint as head and the
other as tail. The oriented incidence matrix $D\in\mathbb{R}^{n\times m}$ is a $(-1,0,1)$-matrix whose rows are indexed by the nodes and columns are indexed by the edges, where the $(i,k)$th entry is $1$ if node $i$ is the head of $e_k$, $-1$ if node $i$ is the tail of $e_k$, and $0$ otherwise. The signed Laplacian $L$ can then be factorized as
\begin{align}
L=DWD'.\label{factorization}
\end{align}
While $D$ depends on the choice of orientations, $L$ does not.

{\subsection{Small disturbance angle stability of power systems}}
Consider a power network with both synchronous generators and inverter-based generators that exploit renewable energy sources. The interconnection of different generators and loads in the power network can be described by an undirected graph $\mathbb{G}\!=\!(\mathcal{V},\mathcal{E})$ consisting of $n$ nodes and $m$ edges, wherein each node represents a bus and each edge represents a transmission line between two buses. {Assume that the transmission lines are lossless.} We use set $\mathcal{V}_1=\{1,2,\dots,n_1
\}$ to represent all the buses with synchronous generators, and set $\mathcal{V}_2=\{n_1+1,n_1+2,\dots,n\}$ to represent the remaining buses with inverter-based generators or frequency-dependent loads. Let $V_i$ and $\theta_i$ be the voltage magnitude and phase angle of bus $i$. {Also, denote by $-B_{ij}$ the susceptance of transmission line $(i,j)\in\mathcal{E}$. Assume $B_{ij}\!=\!B_{ji}$.} The dynamics of phase angle $\theta_i$ at bus $i$ are given by
\begin{equation*}
\begin{split}
&\dot{\theta}_i=\omega_i,\\
&m_i\dot{\omega}_i+k_i\omega_i+\sum_{(i,j)\in\mathcal{E}}V_iV_jB_{ij}\sin(\theta_i-\theta_j)=p_i,
\end{split}
\end{equation*}
where $m_i\!>\!0$ is the inertia of synchronous generator $i\in\mathcal{V}_1$, and $m_i=0$ for $i\in\mathcal{V}_2$; $k_i\!>\!0$ is the damping coefficient of a synchronous generator for $i\!\in\!\mathcal{V}_1$, and the reciprocal of droop gain of an inverter-based generator or frequency dependence coefficient of a load for $i\in\mathcal{V}_2$; $p_i>0$ denotes the generation at a generator bus, and $p_i<0$ denotes the consumption at a load bus.

Let
$\theta=\begin{bmatrix}\theta_1&\theta_2&\dots&\theta_{n}\end{bmatrix}',\omega=\begin{bmatrix}\omega_1&\omega_2&\dots&\omega_{n_1}\end{bmatrix}'$.
Suppose $\begin{bmatrix}\theta^*\\0\end{bmatrix}$ is an equilibrium point. Then $\begin{bmatrix}\theta^*\\0\end{bmatrix}$ is stable if and only if the Jacobian of the system at this equilibrium point has all the eigenvalues in the open left half plane except for a simple zero eigenvalue \cite{chiang1989closest}. The simple zero eigenvalue has a corresponding eigenvector $\begin{bmatrix}\mathbf{1}_n\\0\end{bmatrix}$, meaning $\begin{bmatrix}\theta^*+k\mathbf{1}_n\\0\end{bmatrix},k\in\mathbb{R}$ represents the same equilibrium point as $\begin{bmatrix}\theta^*\\0\end{bmatrix}$.

{In relation to the study of this paper, it is known that semi-stability of the Jacobian matrix amounts to a Laplacian matrix
$L_{\theta^*}\!=\!DW_{\theta^*}D'$ being positive semidefinite with a simple zero eigenvalue, where $D$ is the incidence matrix of $\mathbb{G}$,
$W_{\theta^*}\!=\!\mathrm{diag}\{w_1,w_2,\dots,w_m\}$ and
$w_k=V_iV_jB_{ij}\cos(\theta^*_i-\theta^*_j)$ for $(i,j)\!=\!e_k$. This is mostly stated in the literature under the assumption that $L_{\theta^*}$ has a simple zero eigenvalue; see for instance \cite{chiang1987foundations,chiang1989closest}. We mention without a detailed proof that the argument is also valid when such an assumption is removed.

In most cases, the transmission lines are inductive \mbox{($B_{ij}\!>\!0$)} and $\theta^*_i-\theta^*_j\!<\!\pi/2$ for all $(i,j)\in\mathcal{E}$. Therefore, all the weights $w_k$ are positive and the semidefiniteness requirement of $L_{\theta^*}$ is automatically satisfied provided that $\mathbb{G}$ is connected. However, there do exist real scenarios where either some lines are capacitive ($B_{ij}\!<\!0$) or the angle differences across some lines exceed $\pi/2$ \cite{motter2013spontaneous,song2017network,ding2017impact}. In either case, the corresponding weights $w_k$ are negative and $L_{\theta^*}$ becomes a signed Laplacian. Hence, characterizing the negative weights under which $L_{\theta^*}$ is positive semidefinite with corank 1 is important in studying the small disturbance angle stability.

The semidefiniteness of $L_{\theta^*}$ is not the only thing that is of interest. In case of an unstable equilibrium point, the inertia of $L_{\theta^*}$ determines the type of the equilibrium point (the number of eigenvalues of the Jacobian matrix with positive real parts), which plays a crucial role in the transient stability analysis of power systems \cite{song2017network,chiang1987foundations}.}

\vspace{10pt}

{\subsection{Feasibility of DC power flow}
Consider again a power network of $n$ buses interconnected by $m$ transmission lines, the topology of which is modeled by $\mathbb{G}\!=\!(\mathcal{V},\mathcal{E})$. Suppose that $\mathbb{G}$ is connected.
The DC power flow is a linearized approximation of the AC power flow, assuming lossless transmission lines, unit voltage magnitudes, and small angle differences. It has been widely utilized in the operations of power systems due to its simplicity and easy computation. The DC power flow equations are as follows:
\begin{align*}
p_i=\sum_{(i,j)\in\mathcal{E}}B_{ij}(\theta_i-\theta_j), \;\;\;i=1,2,\dots,n,
\end{align*}
where we use the same notation as in the previous subsection.
Let $p=\begin{bmatrix}p_1&p_2&\dots&p_n\end{bmatrix}',\theta=\begin{bmatrix}\theta_1&\theta_2&\dots&\theta_n\end{bmatrix}'$. Then the equations can be rewritten into the compact form $p=L\theta$, where $L$ is a Laplacian induced by the weights $B_{ij}$. To ensure the feasibility of DC power flow, $L$ has to have a simple zero eigenvalue. This is indeed the case in the majority of scenarios where $B_{ij}>0$ for all $(i,j)\in\mathcal{E}$.
However, as remarked before, there also exist cases, although not common, where $B_{ij}<0$ for some $(i,j)$. In this case, $L$ is a signed Laplacian and may have multiple zero eigenvalues, rendering the DC power flow infeasible. This has been discussed recently based on analysis of semidefiniteless of $L$ and experimental simulations in \cite{ding2017impact}.
Note that the inertia of $L$ is most essential here. In this paper, we will conduct a detailed study on the inertia of $L$, and thus deepen and generalize the analysis in \cite{ding2017impact}.}

{\subsection{Consensus of multi-agent systems with repelling interactions}}
Consider a multi-agent system consisting of $n$ agents, each of which is modeled as a one-dimensional single integrator
\begin{align*}
\dot{x}_i(t)=u_i(t),i=1,2,\dots,n,
\end{align*}
and interacts with its neighbours over an undirected graph $\mathbb{G}=(\mathcal{V},\mathcal{E})$ via the protocol
\begin{align}
u_i(t)=\sum_{(i,j)\in\mathcal{E}}a_{ij}(x_j-x_i),\label{consensusp}
\end{align}
where $a_{ij}\in\mathbb{R}$. Let $x=\begin{bmatrix}x_1&x_2&\dots&x_n\end{bmatrix}'$. The dynamics of the agents can be re-written into the compact form
$\dot{x}(t)=-Lx(t)$,
where $L$ is the Laplacian matrix associated with $\mathbb{G}$.
The agents are said to reach consensus if their states converge to the same value as time goes to infinity.

When all the weights in $\mathbb{G}$ are positive, it is well known that the agents will reach consensus if and only if $\mathbb{G}$ is connected \cite{Murray2004}. This is due to the property that $L\geq 0$ and has a simple zero eigenvalue with a corresponding eigenvector $\mathbf{1}$ if and only if $\mathbb{G}$ is connected. In addition, if the initial condition $x(0)$ is nonnegative, then $x(t)$ will remain nonnegative for all time. This is because $-L$ is exponentially positive when the graph is connected.

When some of the weights are negative due to the possibly repelling interactions among the agents
 or adversarial attacks on the network, $L$ becomes a signed Laplacian. Nevertheless, by the knowledge of stability of linear systems, consensus can still be reached if and only if the signed Laplacian $L\geq 0$ and has a simple zero eigenvalue. {On the other hand, the
 authors in \cite{altafini2015predictable} showed that when $-L$ is eventually exponentially positive, consensus can be reached. We postpone the formal definition of eventual positivity to Section VI, but raise at this stage the following questions: Is the eventual exponential positivity of $-L$ also necessary for guaranteeing consensus? If so, does it suggest certain equivalence between the semidefiniteness of $L$ and the eventual exponential positivity of $-L$? It turns out that the answers to both questions are ``yes''. The detailed reasoning will be justified in Section VI.

We would like to mention that a different consensus protocol was proposed in \cite{altafini} for multi-agent systems with antagonistic relations modeled by signed graphs:
\begin{align*}
u_i(t)=\sum_{(i,j)\in\mathcal{E}}|a_{ij}|(\mathrm{sgn}(a_{ij})x_j-x_i).
\end{align*}
Numerous works have been reported recently in relation to this protocol; see for instance\cite{proskurnikov2016opinion,xia2016,zhang2017bipartite,liu2017expo,liu2018polarizability}. Under this protocol, the dynamics of the coupled agents can be written compactly into
$\dot{x}(t)=-\mathcal{L}x(t)$,}
where $\mathcal{L}=[\ell_{ij}]$ is a different ``signed Laplacian matrix'' defined as
\begin{align*}
\ell_{ij}=\begin{cases}-a_{ij},&i\neq j,\\ \sum_{j=1,j\neq i}^n |a_{ij}|,& i=j.\end{cases}
\end{align*}
Unlike the definition of $L$ adopted in this paper, each diagonal element of $\mathcal{L}$ is equal to the absolute sums of the off-diagonal elements in that particular row. The spectral properties of $\mathcal{L}$ have much to do with the notion of balance of a signed graph. {It has been shown in \cite{altafini} that under such a protocol, the agents either reach bipartite consensus (some of the agents converge to a common value while others converge to the opposite value) or converge trivially to zero, depending on whether the signed graph is balanced or not.}

Comparing two definitions, one can see that
$\mathcal{L}\!=\!L+H$,
where $H$ is a diagonal matrix with diagonal entries being
$h_{ii}\!=\!2 |\sum_{j=1,j\neq i}^n \min(a_{ij},0)|$.
If we let $\mathbb{G}_{\mathrm{loopy}}$ be a graph obtained from $\mathbb{G}$ by adding self-loops to each node $i$ with weight $h_{ii}$, then $\mathcal{L}$ can be regarded as a loopy signed Laplacian associated with $\mathbb{G}_{\mathrm{loopy}}$. Hence, $\mathcal{L}$ extends the loopy Laplacian introduced in \cite{Florian2013} to signed graphs with self-loops. We note that in the recent paper \cite{SAB}, $L$ and $\mathcal{L}$ were called ``repelling Laplacian'' and ``opposing Laplacian'' respectively.

\vspace{12pt}

\section{{ Preliminaries: Graphs as Electrical Networks}} \label{prel}
For a connected signed graph, one can associate with each edge a resistor whose conductance (possibly negative\footnote{Negative conductance corresponds to active resistor that produces power.}) is given by the edge weight. Such an association of signed graphs with resistive networks endows the signed Laplacian $L$ with a nice physical interpretation. Let $\bfc \in\mathbb{R}^n$ be a vector whose elements denote the amount of current injected into each node by external independent sources. Assume
that the sum of the elements of $\bfc$ is zero, i.e., $\bfc'\mathbf{1}=0$, meaning that there is no current accumulating in the electrical
network. Let $\bfu\in\mathbb{R}^n$ be the vector of resulting electric potentials at the nodes. Ohm's law~\cite{laws} says that the current flowing through each edge is equal to the potential difference multiplied by the conductance. Hence, the vector of currents flowing through all the edges is given by $WD'\bfu$. Furthermore, Kirchhoff's current law~\cite{laws} asserts that the difference between the outgoing currents and incoming currents through the edges adjacent to a node is equal to the external current injection at that node. This leads to the current balance equation
$DWD'\bfu\!=\!\bfc$.
In view of (\ref{factorization}), this equation can be rewritten as
\begin{align}
L\bfu=\bfc.\label{currentbalance}
\end{align}
This means that $L$ captures the linear relationship between the current injections and the resulting electric potentials at the nodes. Due to this physical interpretation, $L$ is also referred to as Kirchhoff matrix in some literature \cite{curtis1998circular}.

Now, suppose that one unit of current is injected into node $i$ and extracted from node $j$. Then, the voltage across nodes $i$ and $j$ defines the {\em effective resistance} between nodes $i$ and $j$, denoted by $r_{\mathrm{eff}}(i,j)$. Also, the voltage across nodes $s$ and $t$ defines the {\em transfer effective resistance} between the node pairs $(i,j)$ and $(s,t)$, denoted by $r_{\mathrm{tran}}((s,t),(i,j))$.

The effective resistance and transfer effective resistance can be obtained through experimental measurements. They are also related to the graph Laplacian $L$ as described below.
When $L$ has a simple zero eigenvalue, solving equation (\ref{currentbalance}) yields
\begin{align*}
\bfu=L^{\dagger}\bfc + a\mathbf{1},
\end{align*}
where $L^{\dagger}$ is the Moore-Penrose pseudoinverse of $L$ and $a$ is an arbitrary real number.
Let $\bfc=d_{ij}$. By definition,
\begin{align}
r_{\mathrm{eff}}(i,j)&=d'_{ij}L^{\dagger}d_{ij},\label{eff}\\
r_{\mathrm{tran}}((s,t),(i,j))&=d_{st}'L^{\dagger}d_{ij}.\label{tran}
\end{align}
Since $L^{\dagger}$ is symmetric, we have
\begin{align*}
r_{\mathrm{tran}}((s,t),(i,j))=r_{\mathrm{tran}}((i,j),(s,t)).
\end{align*}
When all the edge weights are positive, it has been shown that the effective resistance serves as a distance function in the node set of a weighted graph\cite{KleRan93}. Recently, some preliminary results on distributed computation of effective resistances have been reported in \cite{Aybat2017DecentralizedCO}.

Note that $r_{\mathrm{eff}}(i,j)$ can be defined for arbitrary two nodes $i$ and $j$, regardless of whether they are connected by an edge or not. Similarly, $r_{\mathrm{tran}}((s,t),(i,j))$ can be defined for arbitrary two node pairs $(i,j)$ and $(s,t)$. If nodes $i$ and $j$ happen to be the head and tail of an edge $e_k$, then $d_{ij}$ coincides with the $k$th column of the incidence matrix $D$. In light of (\ref{eff}) and (\ref{tran}), the effective resistances of the edges along with the transfer effective resistances between different edges can be expressed in a compact matrix form
$D'L^{\dagger}D$.

Using the terminology prevalent in circuit theory \cite{BaoLee07}, we say a resistive network is passive if $\bfu'\bfc \!\geq \!0$, and strictly passive if $\bfu'\bfc \!>\! 0$ for all $\bfc'\mathbf{1}\!=\!0,\bfc\!\neq \!0$.
It is well known that a resistive network with only positive resistances is passive, and is strictly passive if, in addition, the underlying graph is connected. Since the resistive network associated with a signed graph may have both positive and negative resistances, a natural question is: How can we characterize the passivity of the network in terms of the negative resistances? In view of (\ref{currentbalance}), this is a rephrasing of the first question raised in Section II.A using the language of circuits.

\subsection{The Kron reduction}
Consider a resistive electrical network of $n$ nodes with the associated Kirchhoff matrix (signed Laplacian matrix) given by $L$. In many applications, it may happen that only the current balance on a subset of nodes is of importance. These nodes are called external terminals, and the remaning nodes are called interior terminals. Denote by $\alpha\subsetneq\{1,2,\dots,n\}$, $|\alpha|\geq 2$, the set of external terminals, and by $\beta=\{1,2,\dots,n\}\backslash \alpha$ the set of interior terminals. By an appropriate labeling of nodes, we can make the first $|\alpha|$ rows of $L$ correspond to those external terminals. Then, $L$ admits the partition
\begin{align}
L=\begin{bmatrix}L_{\alpha\alpha}&L_{\alpha\beta}\\L_{\beta\alpha}&L_{\beta\beta}\end{bmatrix}.\label{kronpartition}
\end{align}
Consequently, the current balance equation (\ref{currentbalance}) can be rewritten as
\begin{align*}
\begin{bmatrix}\bfc_\alpha\\\bfc_\beta\end{bmatrix}=\begin{bmatrix}L_{\alpha\alpha}&L_{\alpha\beta}\\L_{\beta\alpha}&L_{\beta\beta}\end{bmatrix}\begin{bmatrix}\bfu_\alpha\\\bfu_\beta\end{bmatrix}.
\end{align*}
Applying Gaussian elimination on the interior voltages $\bfu_\beta$, and letting the current injections into the interior terminals $\bfc_\beta\!=\!0$, yields
$\bfc_\alpha=L_{\mathrm{r}}\bfu_\alpha$,
where
\begin{align}
L_{\mathrm{r}}=L_{\alpha\alpha}-L_{\alpha\beta}L^{\dagger}_{\beta\beta}L_{\beta\alpha}.\label{Lred}
\end{align}
One can see that $L_{\mathrm{r}}$ is simply the Schur complement of $L_{\beta\beta}$ in $L$. It has been shown in \cite{Florian2013} that $L_{\mathrm{r}}$ is also a Laplacian and, therefore, is called a reduced signed Laplacian. A reduced signed Laplacian $L_{\mathrm{r}}$ corresponds to a reduced signed graph $\mathbb{G}_{\mathrm{r}}$ involving only all the external terminals. Such a process of obtaining a lower-dimensional reduced network which has the same current balance equation on the external terminals as the original one is called Kron reduction. The Kron reduction process preserves connectivity, i.e., $\mathbb{G}_{\mathrm{r}}$ is connected if and only if $\mathbb{G}$ is connected \cite{Florian2013}.

The Kron reduction can be regarded as a form of abstraction of electrical networks by considering only the current balance at the external terminals. Hereinafter, by saying we perform the Kron reduction on a signed graph $\mathbb{G}$, we mean performing the Kron reduction on the resistive network \mbox{associated with $\mathbb{G}$}.

\subsection{$n$-port network}
Another frequently used abstraction of electrical networks is $n$-port (multiport) networks, which enables capturing behaviours at selected pairs of external terminals while putting the internal structure of the network into a black box. Some preliminaries on $n$-port network theory are given below. See \cite{Anderson73} and the references therein for more details.

As depicted in {Figure \ref{nport}}, an $n$-port network is an electrical network whose external terminals are grouped into $n$ pairs such that for every pair of terminals, the current flowing into one terminal equals the current flowing out of the other. Such pairs of terminals are called ports. Note that different ports are allowed to share a common terminal.
The external behavior of an $n$-port network is completely determined by the port voltages $\bfv_1,\bfv_2,\dots,\bfv_n$, and port currents $\bfi_1,\bfi_2,\dots,\bfi_n$.
Let $\bfv=\begin{bmatrix}\bfv_1&\bfv_2&\dots&\bfv_n\end{bmatrix}'$ and $\bfi=\begin{bmatrix}\bfi_1&\bfi_2&\dots&\bfi_n\end{bmatrix}'$.

\begin{figure}[h!]
\begin{center}
\includegraphics[scale=0.47]{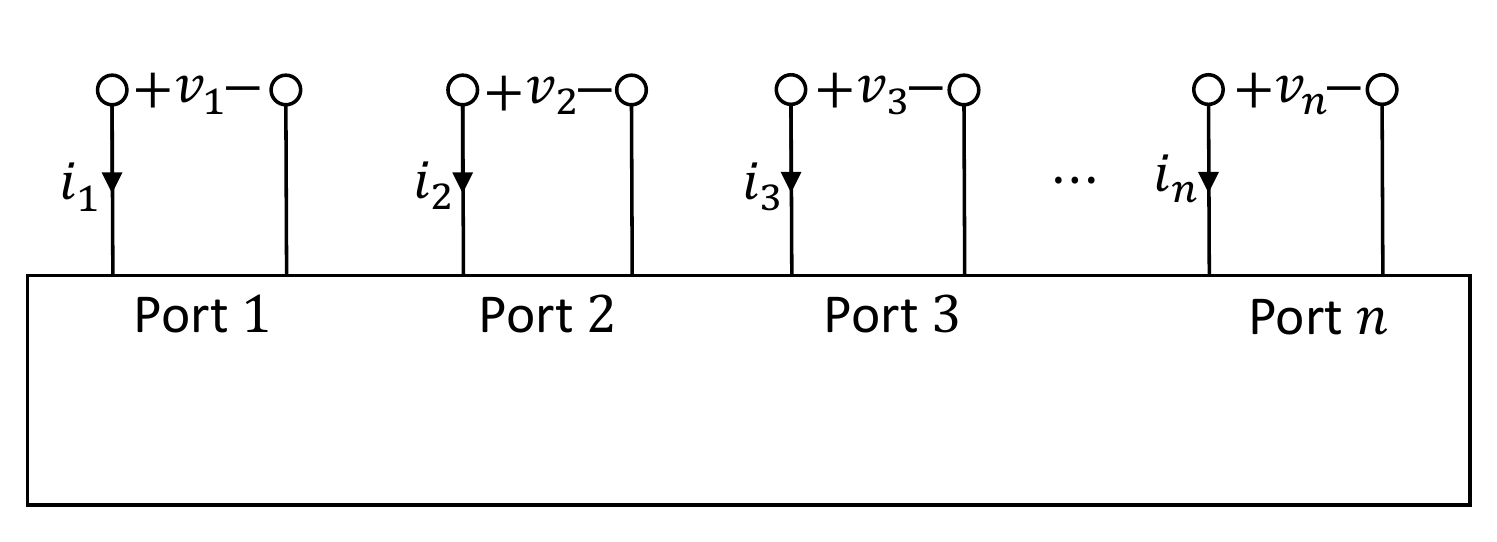}
\vspace{-5pt}
\caption{An $n$-port network.}\label{nport}
\vspace{-5pt}
\end{center}
\end{figure}

Given a resistive $n$-port network, we can express the port voltages in terms of the port currents as $\bfv=Z\bfi$, or express the port currents in terms of the port voltages as $\bfi=Y\bfv$. The matrices $Z$ and $Y$ are symmetric and called the resistance matrix and conductance matrix, respectively.
Each diagonal element of matrix $Z$ ($Y$) represents the resistance (conductance) over a port. Each off-diagonal element of matrix $Z$ ($Y$) represents the transfer resistance (transfer conductance) from one port to another. An $n$-port resistive network is strictly passive if and only if $\bfv'\bfi\!>\!0$ for all nonzero $\bfi$. Therefore, both the resistance matrix and the conductance matrix of a strictly passive $n$-port network are positive definite.

For the resistive network associated with a signed graph, one can select pairs of external terminals as needed and treat the whole network as a multiport. In that case, the resistance over a port is nothing but the effective resistance between the two terminals of the port. The transfer resistance from one port to another coincides with the notion of transfer effective resistance introduced before.

{Note that the resistance matrix $Z$ may not be well-defined because of possibly infinite resistance across certain ports. The infinite resistance may occur due to the dis-connectivity of the network or the presence of negative resistances.} Nevertheless, the conductance matrix $Y$ is always well-defined, since infinite resistance simply means zero conductance. When both $Z$ and $Y$ are well-defined, they are the inverse of each other.

In what follows, we introduce several connections between $n$-port networks and their associated matrix operations.
For an $n$-port network \textit{\textbf{A}}, we use superscript $a$ when denoting its corresponding quantities to differentiate from those of other networks, e.g., $\bfi^a,\bfv^a,Z^a,Y^a$.

First, given an $n$-port network \textit{\textbf{A}}, we can partition the ports into two groups: the first $r$ ports and the remaining $(n-r)$ ports. The resistance matrix and conductance matrix of \textit{\textbf{A}} can be partitioned with compatible dimensions as
\begin{align*}
Z^a=\begin{bmatrix}Z^a_{11}&Z^a_{12}\\Z^a_{21}&Z^a_{22}\end{bmatrix},\quad Y^a=\begin{bmatrix}Y^a_{11}&Y^a_{12}\\Y^a_{21}&Y^a_{22}\end{bmatrix}.
\end{align*}
If we leave the first $r$ ports of network \textit{\textbf{A}} open circuited, i.e., forcing $\bfi^a_{1}=\bfi^a_{2}=\dots=\bfi^a_{r}=0$, we end up with an $(n\!-r)$-port network $\textit{\textbf{C}}$ whose resistance and conductance matrix are
\begin{align}
Z^c=Z^a_{22}\text { and } Y^c=Y^a\slash_{11}.\label{opencircuit}
\end{align}
Instead, if we short the first $r$ ports of \textit{\textbf{A}}, i.e., forcing $\bfv^a_{1}=\bfv^a_{2}=\dots=\bfv^a_{r}=0$, we will end up with another $(n-r)$-port network \textit{\textbf{D}} whose resistance and conductance matrix are
\begin{align}
Z^d=Z^a\slash_{11}\text { and } Y^d=Y^a_{22}.\label{shortconnection}
\end{align}

Now, consider two $n$-port networks \textit{\textbf{A}} and \textit{\textbf{B}}. Let \textit{\textbf{E}} be an $n$-port network obtained by a parallel connection of \textit{\textbf{A}} and \textit{\textbf{B}} as shown in {Figure \ref{parallel}}. Then, the resistance and conductance matrix of \textit{\textbf{E}} are related to those of \textit{\textbf{A}} and  \textit{\textbf{B}} by
\begin{align}
Z^e=\left({Z^a}^\dagger+{Z^b}^\dagger\right)^{\dagger} \text{ and }Y^e=Y^a+Y^b.\label{parallelconnection}
\end{align}

\begin{figure}[htbp]
\begin{center}
\includegraphics[scale=0.45]{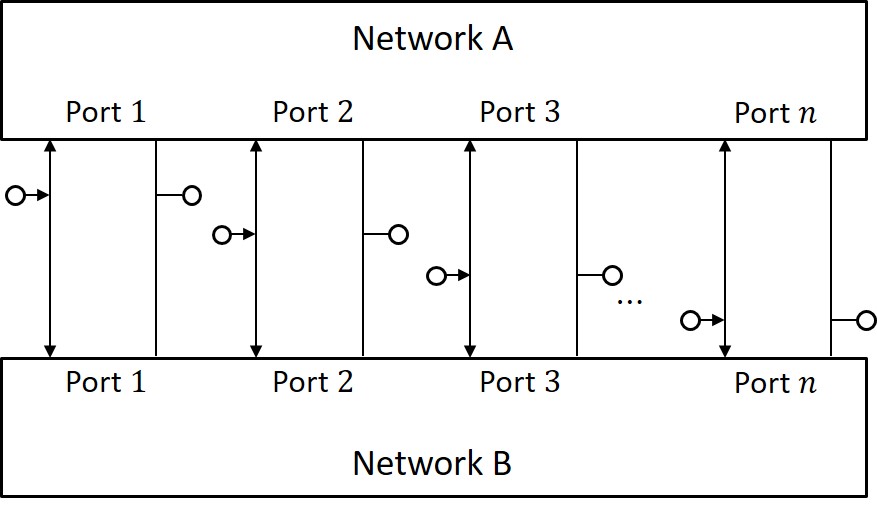}
\caption{Parallel connection of $n$-port networks \textit{\textbf{A}} and \textit{\textbf{B}}.}
\label{parallel}
\end{center}
\end{figure}


\section{Semidefinite Signed Laplacians} \label{semidefiniteness}
\subsection{Semidefiniteness and Kron reduction}
Consider a signed graph $\mathbb{G}\!=\!(\mathcal{V},\mathcal{E})$ with the corresponding signed Laplacian $L$. We treat the nodes incident to negatively weighted edges as external terminals and the remaining nodes as interior terminals. Denote the set of external terminals and the set of interior terminals by $\alpha$ and $\beta$, respectively.

Applying the Kron reduction on $\mathbb{G}$ yields a reduced signed graph $\mathbb{G}_{\mathrm{r}}$ and an associated reduced signed Laplacian $L_{\mathrm{r}}$ as in (\ref{Lred}). Since the Kron reduction preserves connectivity, $\mathbb{G}_{\mathrm{r}}$ is connected if and only if $\mathbb{G}$ is connected. The edges connecting the external terminals in $\mathbb{G}$ remain in $\mathbb{G}_{\mathrm{r}}$. However, many new edges emerge in $\mathbb{G}_{\mathrm{r}}$ due to the reduction process.

In applications, negative weights often appear due to disturbances, faults, or adversarial attacks. Hence, it is reasonable to expect that the number of external terminals, i.e., $|\alpha|$, is much smaller than the size of the graph and, thus, $L_{\mathrm{r}}$ has a much lower dimension than $L$.

The following theorem establishes a connection between the positive semidefiniteness of $L$ and $L_{\mathrm{r}}$.

\begin{theorem}\label{pskron}
For a given signed Laplacian $L\in\mathbb{R}^{n\times n}$, the following statements are equivalent:
\begin{enumerate}[(a)]
\item $L\geq0$ and $\mathrm{corank}(L)=1$;
\item $L_{\mathrm{r}}\geq0$ and $\mathrm{corank}(L_{\mathrm{r}})=1$.
\end{enumerate}
\end{theorem}

This theorem suggests that in dealing with a large network under sparse perturbations on the negatively weighted edges, only behaviors at a small number of external terminals are relevant to the positive semidefiniteness of $L$.

Before formally proving this theorem, we introduce some more notation and useful lemmas.

Denote by $\mathbb{G}_+$ ($\mathbb{G}_-$, respectively) the spanning subgraph of $\mathbb{G}$ with only the positively weighted edges (\mbox{only the negatively} weighted edges, respectively). Also, denote by $L_+$ ($L_-$, respectively) the corresponding Laplacian matrix. Clearly,
\begin{align*}
L=L_+ +L_-.
\end{align*}
As in Section III.A, we relabel the nodes in a way such that the first $|\alpha|$ rows of $L$ correspond to all the external terminals and the rest of the rows correspond to all the interior terminals. Then, $L$ can be partitioned as in (\ref{kronpartition}), and $L_+$ and $L_-$ can be partitioned accordingly as
\begin{align*}
L_+=\begin{bmatrix}L_{+_{\alpha\alpha}}&L_{+_{\alpha\beta}}\\L_{+_{\beta\alpha}}&L_{+_{\beta\beta}}\end{bmatrix},\quad
L_-=\begin{bmatrix}L_{-_{\alpha\alpha}}&0\\0&0\end{bmatrix}.
\end{align*}

\begin{lemma}[\!\!\cite{ChenCDC16}]
\label{mono}
The spectrum of a signed Laplacian $L$ is monotonically increasing with respect to each edge weight of the underlying signed graph $\mathbb{G}$.
\end{lemma}

\begin{lemma}
\label{connectivity}
Given a signed graph $\mathbb{G}$ and its associated signed Laplacian $L$, if $L\geq 0$ with $\mathrm{corank}(L)=1$, then $\mathbb{G}_+$ is connected.
\end{lemma}

\begin{proof}
We prove the lemma by contradiction. Suppose $\mathbb{G}_+$ is not connected; then $L_+$ has multiple zero eigenvalues. {Since $\mathbb{G}=\mathbb{G}_++\mathbb{G}_-$, the edge weights of $\mathbb{G}$ are smaller than or equal to the corresponding weights of $\mathbb{G}_+$. By \mbox{Lemma \ref{mono}}, $\lambda_i(L)\leq \lambda_i(L_+),i=1,2,\dots,n$, where $\lambda_i(L)$ and $\lambda_i(L_+)$ are respectively the $i$th eigenvalues of $L$ and $L_+$, both ordered non-decreasingly.
This means that $L$ must have multiple zero eigenvalues or negative eigenvalues or both. However, $L$ having multiple zero eigenvalues contradicts with $\mathrm{corank}(L)=1$, while $L$ having negative eigenvalues contradicts with $L\geq 0$. This completes the proof.}
\end{proof}

\begin{lemma}
\label{subpo}
Given a signed graph $\mathbb{G}$ and its associated signed Laplacian $L$, if $\mathbb{G}$ is connected, then $L_{\beta\beta}>0$.
\end{lemma}

\begin{proof}
Based on the given signed graph $\mathbb{G}$, we construct a graph $\tilde{\mathbb{G}}$ by negating all the negative edge weights in $\mathbb{G}$ while keeping all the positive edge weights in $\mathbb{G}$ unchanged. Then, $\tilde{\mathbb{G}}$ is a graph with only positive weights. Denote the Laplacian of $\tilde{\mathbb{G}}$ by $\tilde{L}$. Then,
\begin{align*}
\tilde{L}=L_+-L_-=\begin{bmatrix}L_{+_{\alpha\alpha}}-L_{-_{\alpha\alpha}}&L_{+_{\alpha\beta}}\\L_{+_{\beta\alpha}}&L_{+_{\beta\beta}}\end{bmatrix}.
\end{align*}
In view of Corollary 6.2.27 in \cite{HorJoh85}, we have $L_{+_{\beta\beta}}=L_{\beta\beta}>0$ which completes the proof.
\end{proof}

Now we present the proof of Theorem \ref{pskron}.

\textit{Proof of Theorem \ref{pskron}:}
We first show that statement (a) implies statement (b). By Lemma \ref{connectivity}, if $L\!\geq\! 0$ with $\mathrm{corank}(L)\!=\!1$, $\mathbb{G}_+$ is connected and so is $\mathbb{G}$. According to Lemma \ref{subpo}, we have $L_{\beta\beta}>0$. Moreover, there holds
\begin{align}
L=\begin{bmatrix}I&0\\L^{-1}_{\beta\beta}L_{\beta\alpha}&I\end{bmatrix}'\begin{bmatrix}L_{\mathrm{r}}&0\\0&L_{\beta\beta}\end{bmatrix}\begin{bmatrix}I&0\\L^{-1}_{\beta\beta}L_{\beta\alpha}&I\end{bmatrix}.\label{congr}
\end{align}
The statement (b) then follows readily from Sylvester's law of inertia \cite{HorJoh85}.

Now we show the converse direction. Again, by Lemma \ref{connectivity}, if $L_{\mathrm{r}}\!\geq \!0$ with $\mathrm{corank}(L_{\mathrm{r}})\!=\!1$, then the spanning subgraph of $\mathbb{G}_{\mathrm{r}}$ containing all its positively weighted edges, denoted by $\mathbb{G}_{{\mathrm{r}+}}$, is connected. It turns out that $\mathbb{G}_{{\mathrm{r}+}}$ is exactly the reduced graph obtained by applying the Kron reduction on $\mathbb{G}_+$. Since connectivity is preserved by the Kron reduction, $\mathbb{G}_+$ is connected and thus $L_{\beta\beta}>0$. In view of (\ref{congr}), statement (a) follows again from Sylvester's law of inertia.
\hfill\qed

\begin{example}
{Consider the signed graph in Figure \ref{graphtwonegative}, which consists of nine nodes, fifteen positively weighted edges, and two negatively weighted edges.} The weights are labeled on the corresponding edges.
Taking nodes 5, 6, and 7 as external terminals and applying Kron reduction yields a reduced graph as shown in {Figure \ref{reducedgraph}}. Since the reduced graph is connected and does not have any negatively weighted edges, the associated reduced Laplacian is positive semidefinite with a simple zero eigenvalue. Then, according to Theorem \ref{pskron}, we know that the signed Laplacian associated with the signed graph in {Figure \ref{graphtwonegative}} is also positive semidefinite with a simple zero eigenvalue.

\begin{figure}[htbp]
\centering
\includegraphics[scale=0.38]{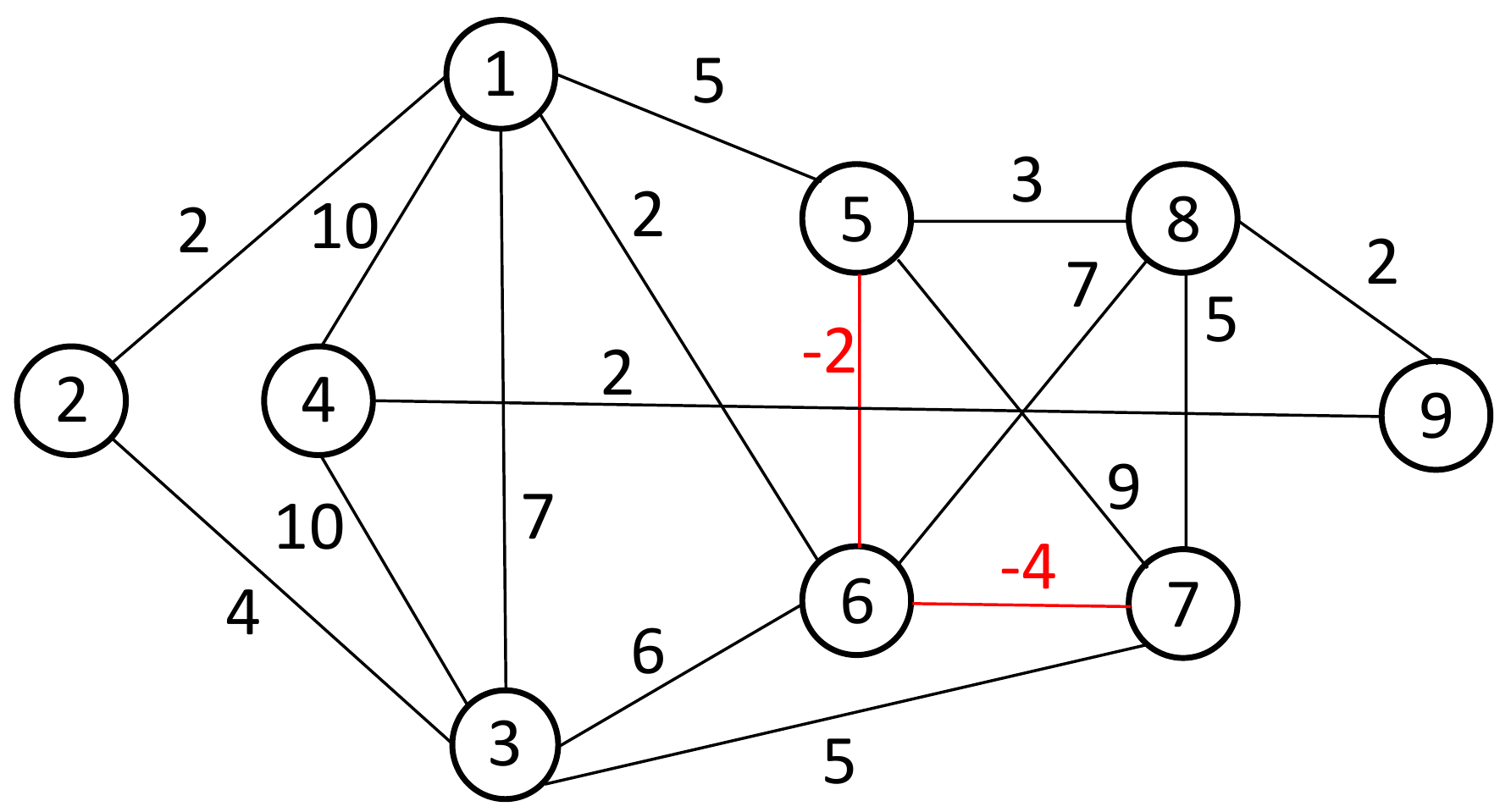}
\caption{A signed graph with two negatively weighted edges.}
\label{graphtwonegative}
\end{figure}
\begin{figure}[htbp]
\centering
\includegraphics[scale=0.38]{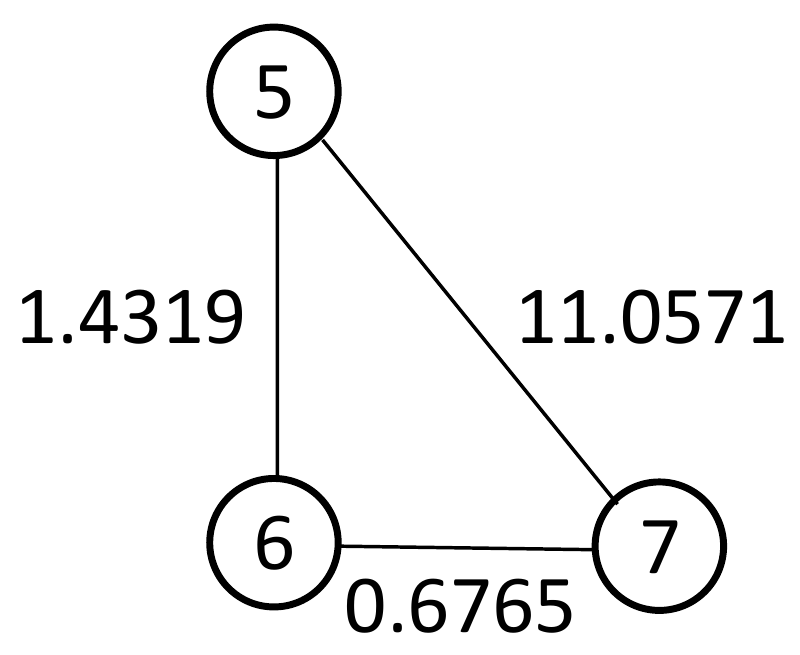}
\caption{The reduced graph after the Kron reduction.}
\label{reducedgraph}
\end{figure}
\label{kronexample}
\end{example}

\subsection{Semidefiniteness and conductance matrix}
In this subsection, we shall characterize the set of negative weights under which a signed Laplacian is positive semidefinite via the notion of resistance matrix or conductance matrix.

{It was shown in \cite{Zelazo} that when a signed graph $\mathbb{G}$ has a single negative edge $(i,j)$, the associated signed Laplacian $L$ is positive semidefinite with a simple zero eigenvalue if and only if $|a_{ij}|< \frac{1}{r^+_{\mathrm{eff}}(i,j)}$, where $r^+_{\mathrm{eff}}(i,j)\!=\!d'_{ij} L^{\dagger}_+ d_{ij}$ is the effective resistance between nodes $i$ and $j$ over the subgraph $\mathbb{G}_+$. This condition can also be expressed as $r_{\mathrm{eff}}(i,j)> 0$ as $\frac{1}{r_{\mathrm{eff}}(i,j)}=\frac{1}{r^+_{\mathrm{eff}}(i,j)}+a_{ij}$.
Our objective is to extend this type of condition to the general case with multiple negative edges and no restrictions on the positions of negative edges.} To this end, we exploit the multiport network theory. Given a signed graph $\mathbb{G}$, we can express $\mathbb{G}$ as the union of three subgraphs:
$\mathbb{G}=\mathbb{F}_-\cup\mathbb{C}_-\cup\mathbb{G}_+$,
where $\mathbb{F}_-=(\mathcal{V},\mathcal{E}_{\mathbb{F}_-})$ is a spanning forest of
$\mathbb{G}_- $, and $\mathbb{C}_-$ is a spanning subgraph of $\mathbb{G}_-$ containing the rest of the edges in $\mathbb{G}_-$.
Suppose that $\mathbb{F}_-$ has $m_{\mathbb{F}_-}$ edges. In many applications, it may well happen that the number of negatively weighted edges is small compared to the size of the whole graph, which means that $m_{\mathbb{F}_-}$ can be much smaller than $n$.

If only the external behavior of the network across the edges in $\mathbb{F}_-$ is the subject of concern, we can treat the two terminals incident to each edge in $\mathbb{F}_-$ as a port, leading to an $m_{\mathbb{F}_-}$-port network. The resistance matrix and conductance matrix of such an $m_{\mathbb{F}_-}$-port network are given by
\begin{align*}
Z_{\mathbb{F}_-}=D_{\mathbb{F}_-}'L^{\dagger}D_{\mathbb{F}_-},\quad\;\; Y_{\mathbb{F}_-}=\left(D_{\mathbb{F}_-}'L^{\dagger}D_{\mathbb{F}_-}\right)^{\dagger},
\end{align*}
respectively, where $D_{\mathbb{F}_-}$ is a submatrix of $D$ that consists of all the columns corresponding to the edges in $\mathbb{F}_-$. Clearly, $D_{\mathbb{F}_-}$ has full rank. Such an $m_{\mathbb{F}_-}$-port network is strictly passive if and only if $Z_{\mathbb{F}_-}>0$, or equivalently, $Y_{\mathbb{F}_-}>0$. In reference to the result in the case of a single negative edge, we raise the following question: Does strict passivity of this $m_{\mathbb{F}_-}$-port network ensure that $L$ is positive semidefinite with corank 1?
The answer is in the affirmative, as captured in the following theorem, {which was shown in the authors' conference paper \cite{chen2017spectral}. The proof exploits the shorted operator of a multiport network and is not included here for brevity.}

\begin{theorem}[\!\!\cite{chen2017spectral}]\label{psresistance}
For a signed Laplacian $L\in\mathbb{R}^{n\times n}$, the following statements are equivalent:
\begin{enumerate}[(a)]
\item $L\geq 0$ with $\mathrm{corank}(L)=1$.
\item $\mathbb{G}_+$ is connected and $Z_{\mathbb{F}_{\!-}}\!>\!0$ (or equivalently, $Y_{\mathbb{F}_-}\!>\!0$).
\end{enumerate}
\end{theorem}

\begin{remark}
When all the weights are positive, the inequality $Z_{\mathbb{F}_-}>0$ becomes irrelevant. Then, Theorem \ref{psresistance} reduces to the well-known result for classical Laplacian matrices.
\end{remark}

The following conclusion can be drawn from Theorem \ref{psresistance}. To characterize the positive semidefiniteness of a signed Laplacian with multiple negative weights, the effective resistances and transfer effective resistances should be considered in a combined way. We note that the choice of a spanning forest $\mathbb{F}_-$ in $\mathbb{G}_-$ is not unique. Nevertheless, Theorem \ref{psresistance} holds for any choice of $\mathbb{F}_-$.

It is important to point out that in Theorem \ref{psresistance} the effects of negative and positive weights on the positive semidefiniteness of $L$ are weaved together in the positive definite matrices $Z_{\mathbb{F}_-}$ and $Y_{\mathbb{F}_-}$. This motivates a further question: Is it possible to have a characterization that explicitly separates the effects of negative and positive weights?

To find answers, we exploit the parallel connection of multiport networks and the associated matrix operation. Specifically, we consider the aforementioned $m_{\mathbb{F}_-}$-port network as a parallel connection of an $m_{\mathbb{F}_-}$-port network with all the positive resistances and another $m_{\mathbb{F}_-}$-port network with all the negative resistances. See {Figure \ref{examplepara}} for an illustration.

\begin{figure}[htbp]
    \centering
    \vspace{-5pt}
    \subfloat[A 2-port network associated with a signed graph, where the edges in black correspond to positive resistances and the edges in red correspond to negative resistances.]{{\includegraphics[scale=0.36]{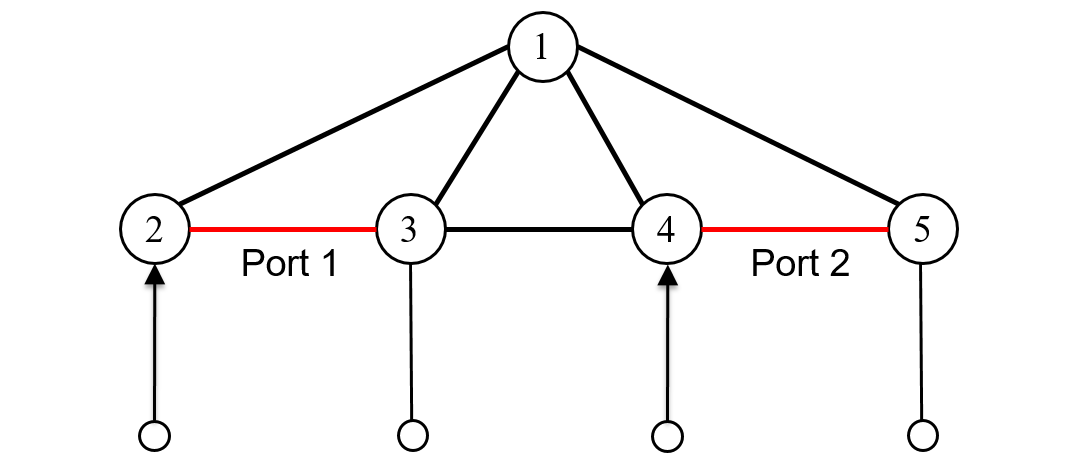} }}
    \vspace{8pt}
    \subfloat[A parallel connection of a 2-port network with positive resistances and a 2-port network with negative resistances.]{{\includegraphics[scale=0.34]{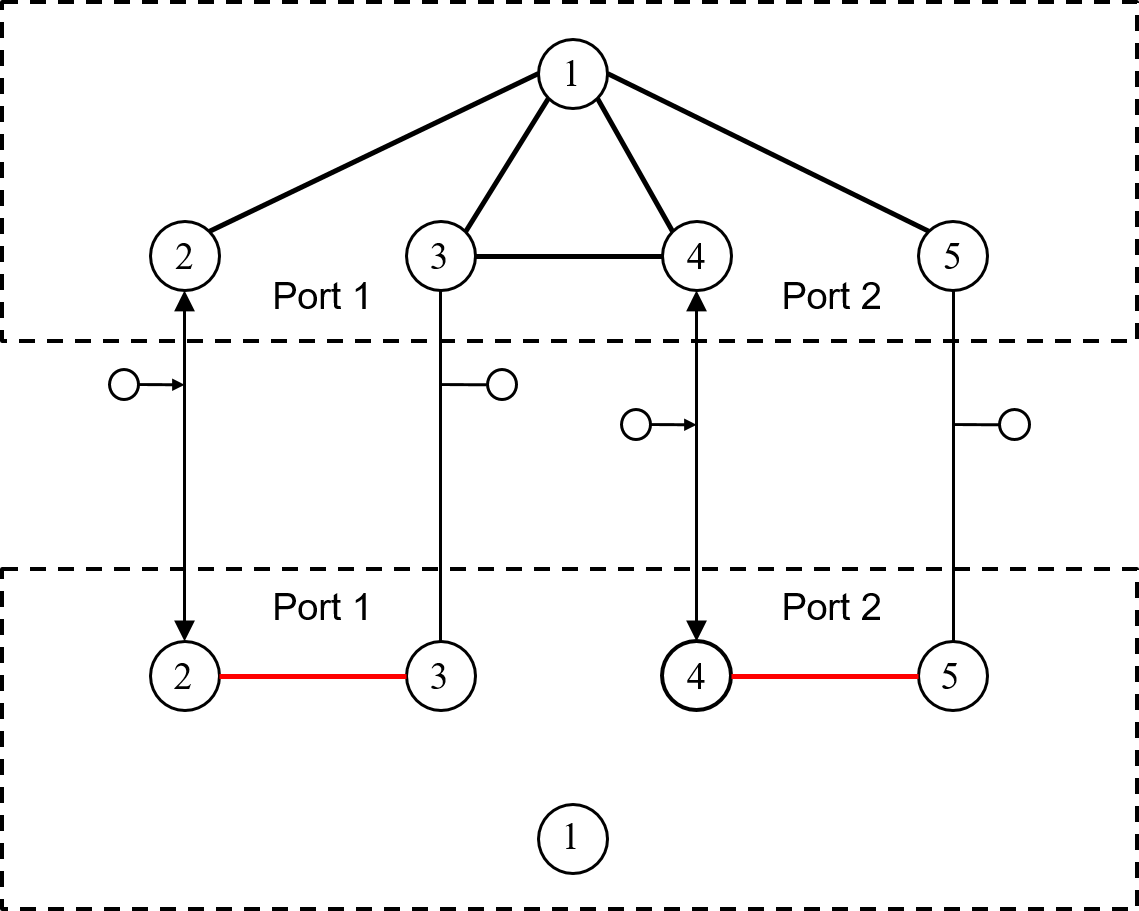} }}
    \caption{A 2-port network considered as a parallel connection of a 2-port with merely positive resistances and a 2-port with merely negative resistances.}
    \label{examplepara}
    \vspace{-10pt}
\end{figure}

Note that the resistance matrices of the $m_{\mathbb{F}_{\!-}}$-port network with only positive resistances and the $m_{\mathbb{F}_-}$-port network with only negative resistances are given by
\begin{align*}
Z^+_{\mathbb{F}_-}=D'_{\mathbb{F}_-}L^{\dagger}_+D_{\mathbb{F}_-}\text{ and }Z^-_{\mathbb{F}_-}=D'_{\mathbb{F}_-}L^{\dagger}_-D_{\mathbb{F}_-},
\end{align*}
respectively. Also, the conductance matrices of the $m_{\mathbb{F}_{\!-}}$-port network with positive resistances and the $m_{\mathbb{F}_-}$-port network with negative resistances are given by
\begin{align*}
Y^+_{\mathbb{F}_-}=\left(D'_{\mathbb{F}_-}L^{\dagger}_+D_{\mathbb{F}_-}\right)^{\dagger}\text{ and }Y^-_{\mathbb{F}_-}=\left(D'_{\mathbb{F}_-}L^{\dagger}_-D_{\mathbb{F}_-}\right)^{\dagger},
\end{align*}
respectively. Then, in view of (\ref{parallelconnection}),
we have
\begin{align*}
Z_{\mathbb{F}_-}=\left({Z_{\mathbb{F}_-}^+}^{\dagger}+{Z_{\mathbb{F}_-}^-}^{\dagger}\right)^{\dagger}\text{ and }Y_{\mathbb{F}_-}=Y^+_{\mathbb{F}_-}+Y^-_{\mathbb{F}_-}.
\end{align*}

The next theorem characterizes the set of negative weights that give rise to a positive semidefinite $L$ with corank 1. {The theorem has been stated in our conference paper \cite{chen2016semidefiniteness}, but without a proof. A concise yet informative proof is provided here by exploiting the parallel connection of multiport networks.} We denote by $D_{\mathbb{G}_-}$ a submatrix of $D$ comprising all the columns of $D$ corresponding to the edges in $\mathbb{G}_-$. Also, we let $W_-$ be the corresponding submatrix of $W$ containing all the negative weights in $\mathbb{G}_-$.

\begin{theorem}\label{resthm3}
For a given signed Laplacian $L\in\mathbb{R}^{n\times n}$, the following statements are equivalent:
\begin{enumerate}[(a)]
\item $L\geq 0$ with $\mathrm{corank}(L)=1$.
\item $\mathbb{G}_+$ is connected, and $Z^+_{\mathbb{F}_-}\!<\!-Z^-_{\mathbb{F}_-}$, or equivalently, $Y^+_{\mathbb{F}_-}\!>\!-Y^-_{\mathbb{F}_-}$.
\end{enumerate}
\end{theorem}

\begin{proof}
One can easily show that when $\mathbb{G}_+$ is connected, $Z^+_{\mathbb{F}_-}=D'_{\mathbb{F}_-}L^{\dagger}_+D_{\mathbb{F}_-}>0$. This is due to the fact that $L^{\dagger}_+\geq 0$, $D_{\mathbb{F}_-}$ has full column rank, and the range of $D_{\mathbb{F}_-}$ is orthogonal to the kernel of $L^{\dagger}_+$. By similar arguments, one can also show $Z^-_{\mathbb{F}_-}<0$.
Hence, in this case,
\begin{align*}
Y^+_{\mathbb{F}_-}={Z^+_{\mathbb{F}_-}}^{-1}>0\text{ and }
Y^-_{\mathbb{F}_-}={Z^-_{\mathbb{F}_-}}^{-1}<0.
\end{align*}
Then, in view of Theorem \ref{psresistance}, to prove the equivalence between (a) and (b), it suffices to that $Y_{\mathbb{F}_-}>0$ is equivalent to $Y^+_{\mathbb{F}_-}>-Y^-_{\mathbb{F}_-}$. This follows directly from $Y_{\mathbb{F}_-}=Y^+_{\mathbb{F}_-}+Y^-_{\mathbb{F}_-}$.
\end{proof}

\begin{remark}
In many applications, $\mathbb{G}_+$ represents a nominal graph which may suffer from perturbations in the form of negatively weighted edges. In this regard, $Y^+_{\mathbb{F}_-}$ can be considered as a measure of fragility of $\mathbb{G}_+$ under such perturbations. The larger $Y^+_{\mathbb{F}_-}$ is, the less fragile $\mathbb{G}_+$ is.
\end{remark}

From the above theorem, we have the following corollary.

\begin{corollary}\label{cor1}
If $\mathbb{G}$ does not have any cycle containing two or more negative edges,
then $L\!\geq\! 0$ with $\mathrm{corank}(L)\!=\!1$ if and only if
$\mathbb{G}_+$ is connected, and $r^+_{\mathrm{eff}}(i,j)\!<\!\frac{1}{|a_{ij}|}$ for all $(i,j)\!\in\!\mathcal{E}_{-}$.
\end{corollary}

\begin{proof}
When $\mathbb{G}$ does not have any cycle containing two or more negatively weighted edges, the voltage across any port $(i,j)\in\mathbb{F}_-$ merely depends on the current through the port itself. In this case, both $Z_{\mathbb{F}_-}^+$ and $Z_{\mathbb{F}_-}^-$ are diagonal matrices with the corresponding diagonal elements given by $r^+_{\mathrm{eff}}(i,j)$ and $\frac{1}{a_{ij}}$, respectively. The conclusion then follows immediately from Theorem \ref{resthm3}.
\end{proof}

Corollary \ref{cor1} is consistent with the statements in \cite[Theorem III.4]{Zelazo} and \cite[Theorem 3.2]{yxchen}.

\begin{example}
We consider a modified version of the signed graph in Example \ref{kronexample}, where the positive weights are labeled on the respective edges as in {Figure \ref{graphtwonegative}}, but the negative weights $a_{56},a_{67}$ are no longer fixed a priori. By Theorem \ref{resthm3}, the set of negative weights that give rise to a positive semidefinite signed Laplacian with a simple zero eigenvalue is characterized by the inequality
\begin{align*}{
\begin{bmatrix}-\frac{1}{a_{56}}&0\\0&-\frac{1}{a_{67}}\end{bmatrix}>\begin{bmatrix}0.1488&-0.1046\\-0.1046&0.1371\end{bmatrix},}
\end{align*}
and depicted as the interior of the shaded area in {Figure \ref{negativeweightset}}.
\begin{figure}[h!]
\centering
\includegraphics[scale=0.35]{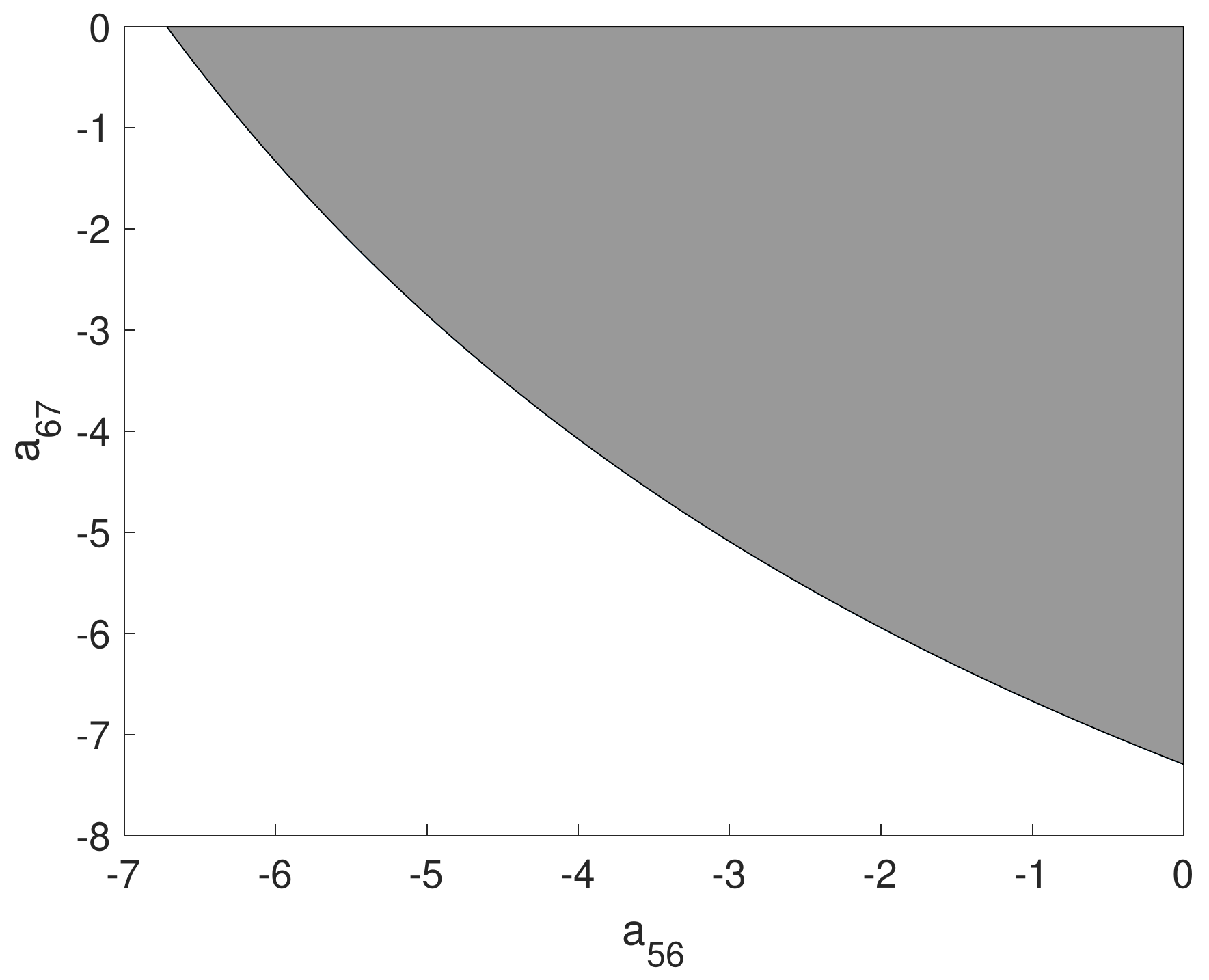}
\caption{The set of negative weights $a_{56},a_{67}$ that give rise to $L\geq 0$ with a simple zero eigenvalue.}
\label{negativeweightset}
\end{figure}
\end{example}

{\begin{remark}
We wish to mention that the semidefiniteness of a signed Laplacian $L$ can also be verified in a sequential way, as if the negative edges are added one by one. The procedure is sketched as below. Starting from a connected $\mathbb{G}_+$, we add one negative edge back, say $(i,j)$, forming a graph $\mathbb{G}_1$ with the associated signed Laplacian $L_1$. We know that if $|a_{ij}|<\frac{1}{d'_{ij} L^{\dagger}_+ d_{ij}}$, then $L_1\!\geq\! 0$ with corank 1. Next, we add another \mbox{negative edge} back to $\mathbb{G}_1$, say $(k,l)$, forming a graph $\mathbb{G}_2$ with the associated signed Laplacian $L_2$. It is not difficult to show that if $|a_{kl}|<\frac{1}{d'_{kl} L^{\dagger}_1 d_{kl}}$, then $L_2\!\geq\! 0$ with corank 1. As a matter of fact, this can be inferred from \cite[Theorem 3]{song2017network}. As this process repeats itself, if the corresponding inequality condition continues to hold until the addition of the last negative edge, then we know that $L\geq 0$ with corank 1. Otherwise, $L$ is indefinite or has multiple zero eigenvalues or both.
\end{remark}}

\section{Inertias of Signed Laplacians}\label{inertia}
{When a signed Laplacian is indefinite, its inertia is often of importance, as discussed in the motivating applications in Section II.} It turns out that the Kron reduced signed Laplacian and the conductance matrix of the multiport network encapsulate the inertia of an indefinite signed Laplacian.

\subsection{Inertia and Kron reduction}
As in Section IV.A, we treat the nodes incident to negatively weighted edges as external terminals and the remaining nodes as interior terminals. Then, the signed Laplacian $L$ admits the partition as in (\ref{kronpartition}).
Applying the Kron reduction on $\mathbb{G}$ leads to a reduced network $\mathbb{G}_{\mathrm{r}}$ with the associated signed Laplacian $L_{\mathrm{r}}$.
\begin{theorem}
Assume that $\mathbb{G}$ is connected. Then
$\pi(L)=\pi(L_{\mathrm{r}})+(0,0,|\beta|)$.
\end{theorem}
\begin{proof}
Since $\mathbb{G}$ is connected, by Lemma \ref{subpo}, $L_{\beta\beta}>0$. In view of (\ref{congr}), it follows from Sylvester's law of inertia that
\begin{align*}
\pi(L)=\pi(L_{\mathrm{r}})+\pi(L_{\beta\beta})=\pi(L_{\mathrm{r}})+(0,0,|\beta|).
\end{align*}
This completes the proof.
\end{proof}

\subsection{Inertia and conductance matrix}
As in Section IV.B, let $\mathbb{F}_-$ be an arbitrary spanning forest of $\mathbb{G}_-$. Taking the two nodes incident to each edge in $\mathbb{F}_-$ as two external terminals of a port gives rise to an $m_{\mathbb{F}_-}$-port network. The following theorem states an explicit relation between the inertia of the signed Laplacian $L$ and that of the conductance matrix $Y_{\mathbb{F}_-}$. {An earlier version was reported under a stronger assumption that $\mathbb{G}_+$ is connected in our conference paper \cite{chen2017spectral}.}
\begin{theorem}
\label{inertiacon}
Assume that $\mathbb{G}$ is connected. Then
$\pi(L)=\pi(Y_{\mathbb{F}_-})+(0,1,n\!-\!1\!-\!m_{\mathbb{F}_-})$.
\end{theorem}

\begin{proof}
Since $\mathbb{G}$ is connected, we can augment $\mathbb{F}_-$ with $n\!-\!1\!-\!m_{\mathbb{F}_-}$ edges from $\mathbb{G}_+$ to form a spanning tree $\mathbb{F}$ of $\mathbb{G}$. Considering the two nodes of each edge in $\mathbb{F}$ as two terminals that constitute a port, we have an augmented $(n-1)$-port network \textit{\textbf{A}} with the conductance matrix $Y^a$. Then,
\begin{align*}
\pi(L)&=\pi(L^{\dagger})=\pi\left(\begin{bmatrix}D_{\mathbb{F}} &\mathbf{1}\end{bmatrix}'L^{\dagger}\begin{bmatrix}D_{\mathbb{F}} &\mathbf{1}\end{bmatrix}\right)\\
&=\pi\left(\begin{bmatrix}D'_{\mathbb{F}}L^{\dagger}D_{\mathbb{F}} &0\\0&0\end{bmatrix}\right)=\pi\left(D'_{\mathbb{F}}L^{\dagger}D_{\mathbb{F}}\right)+(0,1,0)\\
&=\pi(Y^a)+(0,1,0),
\end{align*}
where the second equality is due to Sylvester's law of inertia, and the last equality is due to $Y^a=\left(D'_{\mathbb{F}}L^{\dagger}D_{\mathbb{F}}\right)^{\dagger}$.

We label those ports corresponding to $\mathbb{F}_-$ as the first $m_{\mathbb{F}_-}$ ports. Then, the conductance matrix admits the partition
\begin{align*}
Y^a=\begin{bmatrix}Y^a_{11}&Y^a_{12}\\Y^a_{21}&Y^a_{22} \end{bmatrix}.
\end{align*}
Shorting all the $m_{\mathbb{F}_-}$ ports corresponding to $\mathbb{F}_-$, we obtain a shorted network \textit{\textbf{D}} as depicted in {Figure \ref{shorted}}. In view of (\ref{shortconnection}), the conductance matrix of \textit{\textbf{D}} is given by $Y^a_{22}$.
\begin{figure}[htbp]
\centering
\includegraphics[scale=0.5]{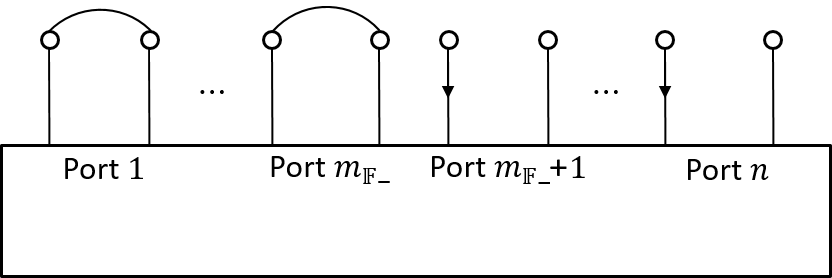}
\caption{A shorted connection \textit{\textbf{D}} of \textit{\textbf{A}}.}
\label{shorted}
\end{figure}

Since $\mathbb{G}$ is connected and no current flows through the negative resistors, the shorted network \textit{\textbf{D}} is strictly passive and, thus, $Y^a_{22}\!>\!0$. Then, applying the inertia additive formula of Schur complement \cite[Theorem 1.6]{zhang2006schur}, we have
\begin{align*}
\pi(Y^a)\!=\!\pi(Y^a\!/_{22})\!+\!\pi(Y^a_{22})\!=\!\pi(Y^a\!/_{22})+(0,0,n\!-\!1\!-\!m_{\mathbb{F}_-}).
\end{align*}
In view of (\ref{opencircuit}), we have $Y^a\!/_{22}=Y_{\mathbb{F}_-}$. Hence,
\begin{align*}
\pi(L)&=\pi(Y^a)+(0,1,0)\\
&=\pi(Y_{\mathbb{F}_-})+(0,0,n-1-m_{\mathbb{F}_-})+(0,1,0)\\
&=\pi(Y_{\mathbb{F}_-})+(0,1,n\!-\!1\!-\!m_{\mathbb{F}_-}),
\end{align*}
which completes the proof.
\end{proof}

The authors in \cite{song2017network} gave an alternative way of characterizing the inertia of a signed Laplacian $L$, under the assumption that $L$ has a simple zero eigenvalue.

{\begin{remark}
With some additional effort, one can deduce from Theorem \ref{inertiacon} the following inertia bounds first reported in \cite[Theorem 2.10]{la1}:
\begin{align*}
c(\mathbb{G}_+)-1\leq &\pi_-(L)\leq n - c(\mathbb{G}_-),\\
c(\mathbb{G}_-)-1\leq &\pi_+(L)\leq n-c(\mathbb{G}_+),\\
1\leq &\pi_0(L)\leq n+2-c(\mathbb{G}_-)-c(\mathbb{G}_+),
\end{align*}
where $c(\mathbb{G}_+)$ and $c(\mathbb{G}_-)$ represent the numbers of connected components in $\mathbb{G}_+$ and $\mathbb{G}_-$, respectively. Take the bounds on $\pi_-(L)$ for an illustration. The upper bound is straightforward as
$\pi_-(L)=\pi_-(Y_{\mathbb{F}_-})\leq m_{\mathbb{F}_-}=n-c(\mathbb{G}_-)$. Regarding the lower bound, note that there must exist $c(\mathbb{G}_+)-1$ number of negatively weighted edges linking the connected components of $\mathbb{G}_+$ together so as to form a connected spanning subgraph of $\mathbb{G}$, denote by $\tilde{\mathbb{G}}$. See {Figure \ref{inertiatopology}} for an illustration. Treating each negatively weighted edge in $\tilde{\mathbb{G}}$ as a port, one can view $\tilde{G}$ as a $(c(\mathbb{G}_+)-1)$-port network, the conductance matrix of which is diagonal with diagonal elements given simply by the negative weights in $\tilde{\mathbb{G}}$. Let $\tilde{L}$ be the signed Laplacian associated with $\tilde{\mathbb{G}}$. From Theorem \ref{inertiacon}, we have
$\pi_-(\tilde{L})=c(\mathbb{G}_+)-1$ and thus $\pi_-(L)\geq \pi_-(\tilde{L})=c(\mathbb{G}_+)-1$.
\end{remark}}

\begin{figure}[htbp]
    \centering
    \subfloat[$\mathbb{G}$]{{\includegraphics[scale=0.6]{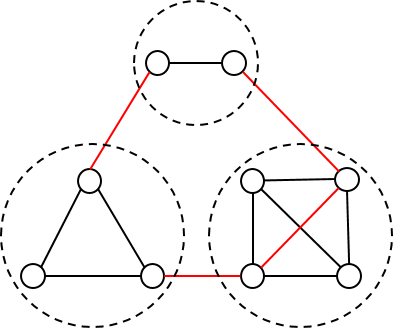} }}%
    \quad
    \subfloat[$\tilde{\mathbb{G}}$]{{\includegraphics[scale=0.6]{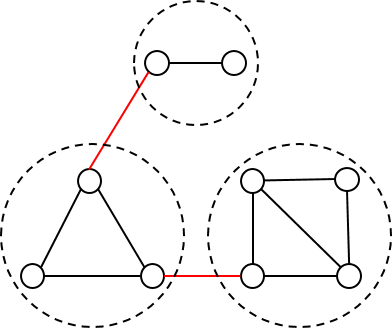} }}%
    \caption{A signed graph $\mathbb{G}$ and a spanning subgraph $\tilde{\mathbb{G}}$, where the positively weighted edges are in black and the negatively weighted edges are in red.}%
    \label{inertiatopology}%
\end{figure}

\begin{example}
We consider a modified version of the signed graph in Example \ref{kronexample}, where the positive weights are as labelled on the respective edges in {Figure \ref{graphtwonegative}} and the negative weights are changed to {$a_{56}=-3,a_{67}=-7$}. Then,
\begin{align*}{
Y_{\mathbb{F}_-}=\begin{bmatrix}11.4890&11.0571\\11.0571&8.7337\end{bmatrix},}
\end{align*}
which has one positive eigenvalue and one negative eigenvalue. It follows from Theorem \ref{inertiacon} that
\begin{align*}
{\pi(L)=(1,0,1)+(0,1,6)=(1,1,7).}
\end{align*}
If we retain the positive weights and increase the magnitudes of the negative weights so that {$a_{56}=-40,a_{67}=-30$}, then
\begin{align*}
{Y_{\mathbb{F}_-}=\begin{bmatrix}-25.5110&11.0571\\11.0571&-14.2663\end{bmatrix},}
\end{align*}
which has two negative eigenvalues. Again, by Theorem \ref{inertiacon}, we have
{$\pi(L)=(2,0,0)+(0,1,6)=(2,1,6)$}.
\end{example}

\section{Signed Laplacians and Eventual Positivity}
{In this section, we address the third question raised in Section II.A. Before proceeding, some preliminaries on eventual positivity and Perron-Frobenius property are introduced.}

It is widely recognized that the nonnegativity (exponential nonnegativity, respectively) of a matrix $A$ indicates the orthant invariance of the system $x(k+1)=Ax(t)$ ($\dot{x}(t)=Ax(t)$, respectively), i.e., the trajectory of the system's state remains in the nonnegative orthant if it starts with a nonnegative initial condition.
Recently, much attention has been paid to the eventual nonnegativity (eventual exponential nonnegativity) of a matrix $A$, which indicates only asymptotic orthant invariance, i.e., given a nonnegative initial condition, the state trajectory could exit the nonnegative orthant temporarily and returns to it at a future time and remains therein forever.
\begin{definition}
A matrix $A\in\mathbb{R}^{n\times n}$ is said to be eventually positive (nonnegative, respectively) if there is a positive integer $k_0$, such that $A^k\rhd 0$ ($A^k\unrhd 0$, respectively) for all $k\geq k_0$.
\end{definition}
\begin{definition}
A matrix $A\in\mathbb{R}^{n\times n}$ is said to be eventually exponentially positive (nonnegative, respectively) if there is a positive real number $t_0$, such that $e^{At}\rhd 0$ ($e^{At}\unrhd 0$, respectively) for all $t\geq t_0$.
\end{definition}

The relationship between eventual positivity and eventual exponential positivity is given below.
\begin{lemma}[\!\!\cite{noutsos2008reachability}]\label{eep}
A matrix $A\in\mathbb{R}^{n\times n}$ is eventually exponentially positive if and only if there exists $s\geq 0$ such that $A+sI$ is eventually positive.
\end{lemma}

As is well known, positive matrices possess the so-called strong Perron-Frobenius property, but the converse is not true. Recently, the equivalence between eventual positivity and the strong Perron-Frobenius property has been established.

\begin{definition}[\!\!\cite{noutsoslaa}]
A matrix $A\in\mathbb{R}^{n\times n}$ is said to possess the strong Perron-Frobenius property if $\rho(A)$ is a simple positive eigenvalue with a positive right eigenvector and $|\lambda|\!<\!\rho(A)$ for every other eigenvalue $\lambda\neq \rho(A)$ of $A$.
\end{definition}

\begin{lemma}[\!\!\cite{noutsoslaa}]\label{PFP}
Let $A\!\in\!\mathbb{R}^{n\times n}$ be a symmetric matrix. Then it possesses the strong Perron-Frobenius property if and only if it is eventually positive.
\end{lemma}

Based on eventual nonnegativity, the following generalization of M-matrix has been introduced in \cite{olesky2009m}.
\begin{definition}
A matrix $A\in\mathbb{R}^{n\times n}$ is an eventual M-matrix if it can be expressed as $A=sI-B$, where $s\geq \rho(B)$ and $B$ is eventually nonnegative.
\end{definition}

Now, we are ready to state the main theorem of this section.
\begin{theorem}\label{theoremep}
For a given signed Laplacian $L\in\mathbb{R}^{n\times n}$, the following statements are equivalent:
\begin{enumerate}[(a)]
\item $L\geq0$ and $\mathrm{corank}(L)=1$.
\item $L$ is an eventual M-matrix that can be expressed as $L=sI-B$, where $s=\rho(B)$, and $B$ is eventually positive.
\item $-L$ is eventually exponentially positive.
\end{enumerate}
\begin{proof}
First, we show that (a) implies (b). Suppose $L\geq0$ and has corank 1. Let $\lambda_1,\lambda_2,\dots,\lambda_n$ be the eigenvalues of $L$ ordered nondecreasingly, i.e.,
$0=\lambda_1<\lambda_2\leq \lambda_3\leq \dots\leq \lambda_n$.
Let $s=\lambda_n$ and $B=sI-L$. Clearly, $\rho(B)=\lambda_n$ is a simple positive eigenvalue of $B$ with a corresponding eigenvector $\mathbf{1}$. Moreover, $\rho(B)$ is greater than the magnitude of any other eigenvalue of $B$. Therefore, $B$ possesses the strong Perron-Frobenius property and is eventually positive by Lemma \ref{PFP}. This validates that $L=sI-B$ is an eventual M-matrix with $B$ being eventually positive.

Second, we show that (b) implies (c). From the fact that $L$ is an eventual M-matrix with $L=sI-B$ and $B$ is eventually positive, it follows that $-L+sI$ is eventually positive. By Lemma \ref{eep}, $-L$ is eventually exponentially positive.

Finally, we show that (c) implies (a). By definition, if $-L$ is eventually exponentially positive, then $e^{-L}$ is eventually positive and, hence, satisfies the strong Perron-Frobenius property in light of Lemma \ref{PFP}. Note that $e^{-L}$ has an eigenvalue $1$ with a corresponding eigenvector $\mathbf{1}$. Since $e^{-L}$ is symmetric, all the other eigenvectors are orthogonal to $\mathbf{1}$ and hence cannot be positive. Therefore, $\rho(e^{-L})=1$ is a simple eigenvalue of $e^{-L}$ and is greater than all the other eigenvalues in absolute value. From this, the statement (a) follows.
\end{proof}
\end{theorem}

The implications of Theorem \ref{theoremep} in linear consensus problems are interesting. {It answers the questions in Section II.D regarding consensus over signed graphs under the protocol (\ref{consensusp}). When there exist repelling interactions between some agents, the corresponding edge weights are negative and $L$ is a signed Laplacian.
It is known that consensus can be reached if and only if $L\geq 0$ with a simple zero eigenvalue. Then, Theorem \ref{theoremep} tells that the eventual exponential positivity of $-L$ is not only sufficient in guaranteeing consensus as indicated in \cite{altafini2015predictable}, but also necessary.} Moreover, $-L$ being eventually exponentially positive means that for a given nonnegative initial condition $x(0)$, the agents' states $x(t)$ may exit the nonnegative orthant temporarily and return to it at a future time and remain therein forever. {This is a prominent distinction from the consensus over a conventional graph with only positive weights, in which the states always stay nonnegative under a nonnegative initial condition.} Below is a simulation example for illustration.

{\begin{example}
Consider a consensus problem of nine agents interacting over the signed graph depicted in Figure \ref{graphtwonegative}, and under the consensus protocol (\ref{consensusp}). From the analysis in Example \ref{kronexample}, we already know that the associated signed Laplacian $L$ is positive semidefinite
with a simple zero eigenvalue. Hence, the agents can reach consensus. Let the initial states be
\begin{align*}
\begin{split}
x(0)= [ \begin{matrix}1.61&3.52&6.02&4.76&17.57&0.1&10.87&2.28&15.79\end{matrix} ]'.
\end{split}
\end{align*}
The simulation result is shown in Figure \ref{trajectory}, where we can see that consensus is indeed reached, but the state of the sixth agent becomes negative temporarily due to the eventual exponential positivity of $-L$.

\begin{figure}[htbp]
\centering
\includegraphics[scale=0.4]{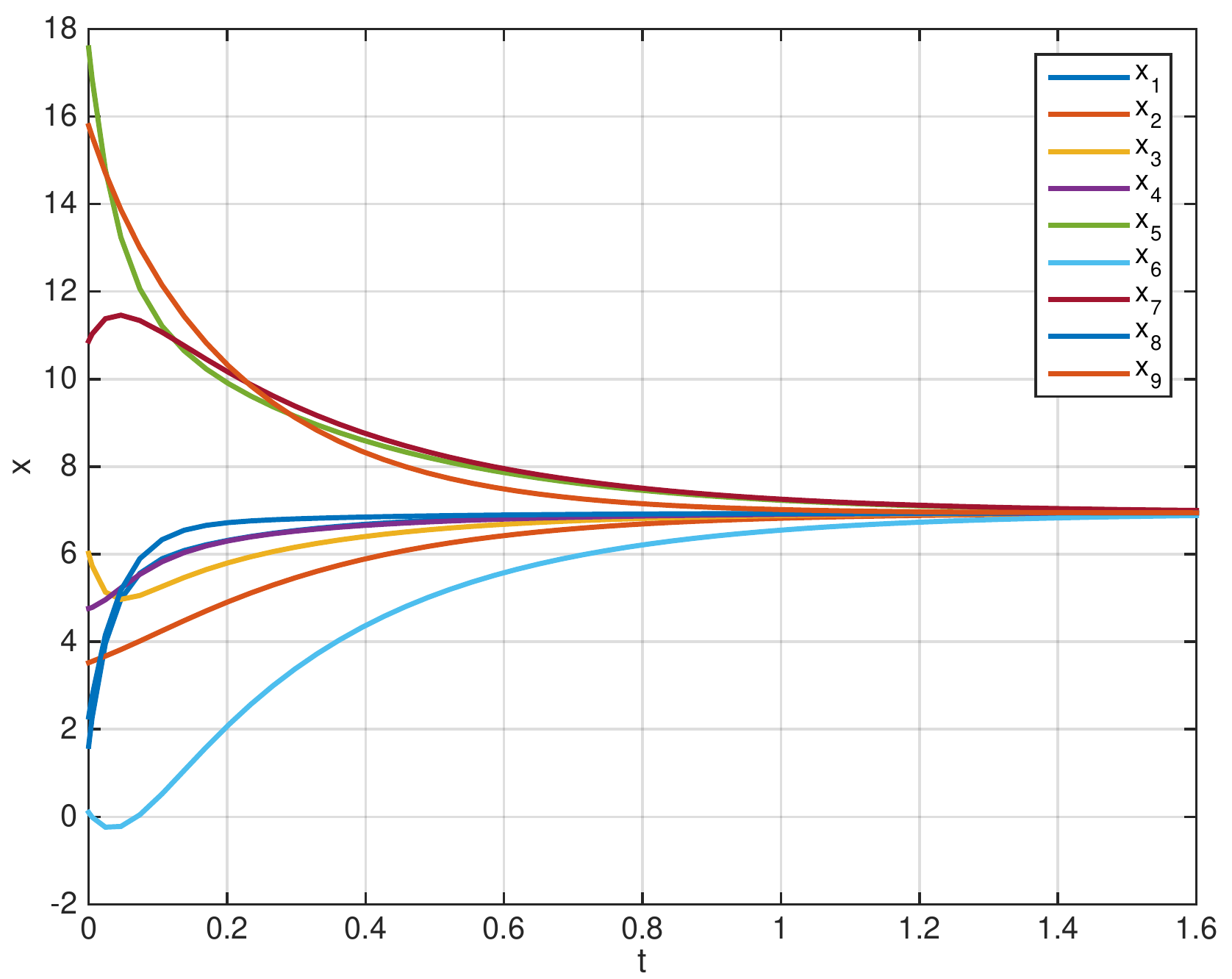}
\caption{State trajectories of the agents.}
\label{trajectory}
\end{figure}
\end{example}}

\section{Conclusion} \label{ending}
In this paper, we studied the spectral properties of signed Laplacians. We first characterized the positive semidefiniteness of signed Laplacians in terms of the negative weights via both the Kron reduction and $n$-port network theory. The study was then extended to characterizing the inertias of indefinite signed Laplacians. Moreover, we revealed the connections between signed Laplacians, generalized M-matrices, and eventually exponentially positive matrices.

One future direction of research is to extend the study to directed signed graphs. We wish to explore conditions on the negative weights under which the associated signed Laplacians have all the eigenvalues in the open right half plane except for a simple zero eigenvalue. Preliminary results on some special cases can be found in \cite{ahmadizadeh2017eigenvalues,ASMNc,mukherjee2016consensus}. A generalization of the notion of effective resistance to directed graphs proposed in \cite{YSL16p1,YSL16p2} may give a clue in this exploration. We also envision that the connection with eventual exponential positivity will continue to play an important role in the directed case.

\section*{Acknowledgements}
The authors would like to thank Prof. Claudio Altafini of Link\"{o}ping University, Prof. Florian D\"{o}rfler of Swiss Federal Institute of Technology, and Dr. Yue Song, Prof. Tao Liu, and Prof. David J. Hill of University of Hong Kong for valuable discussions.


\begin{thebibliography}{10}

\bibitem{heider1946attitudes}
F.~Heider.
\newblock Attitudes and cognitive organization.
\newblock {\em J Psychol.}, 21(1):107--112, 1946.

\bibitem{harary1953notion}
F.~Harary.
\newblock On the notion of balance of a signed graph.
\newblock {\em Michigan Math. J.}, 2(2):143--146, 1953.

\bibitem{taylor1970balance}
H.~F. Taylor.
\newblock {\em Balance in Small Groups}.
\newblock Van Nostrand Reinhold Co, 1970.

\bibitem{altafini}
C.~Altafini.
\newblock Consensus problems on networks with antagonistic interactions.
\newblock {\em IEEE Trans. Autom. Control}, 58(4):935--946, 2013.

\bibitem{altafini2015predictable}
C.~Altafini and G.~Lini.
\newblock Predictable dynamics of opinion forming for networks with
  antagonistic interactions.
\newblock {\em IEEE Trans. Autom. Control}, 60(2):342--357, 2015.

\bibitem{shi2016evolution}
G.~Shi, A.~Proutiere, M.~Johansson, J.~S. Baras, and K.~H. Johansson.
\newblock The evolution of beliefs over signed social networks.
\newblock {\em Operations Research}, 64(3):585--604, 2016.

\bibitem{proskurnikov2016opinion}
A.~V. Proskurnikov, A.~S. Matveev, and M.~Cao.
\newblock Opinion dynamics in social networks with hostile camps: Consensus vs.
  polarization.
\newblock {\em IEEE Trans. Autom. Control}, 61(6):1524--1536, 2016.

\bibitem{liu2017expo}
J.~Liu, X.~Chen, T.~Ba{\c{s}}ar, and M.-A. Belabbas.
\newblock Exponential convergence of the discrete- and continuous-time
  {A}ltafini models.
\newblock {\em IEEE Trans. Autom. Control}, 62(12):6168--6182, 2017.

\bibitem{liu2018polarizability}
F.~Liu, D.~Xue, S.~Hirche, and M.~Buss.
\newblock Polarizability, consensusability, and neutralizability of opinion
  dynamics on coopetitive networks.
\newblock {\em IEEE Trans. Autom. Contr.}, 64(8):3339--3346, 2018.

\bibitem{SAB}
G.~Shi, C.~Altafini, and J.~S. Baras.
\newblock Dynamics over signed networks.
\newblock {\em SIAM Review}, 61:229--257, 2019.

\bibitem{Boyd2004}
L.~Xiao and S.~Boyd.
\newblock Fast linear iterations for distributed averaging.
\newblock {\em Syst. Control Lett.}, 53(1):65--78, 2004.

\bibitem{Zelazo}
D.~Zelazo and M.~B{\"u}rger.
\newblock On the definiteness of the weighted {L}aplacian and its connection to
  effective resistance.
\newblock In {\em Proc. 53rd IEEE Conf. Decision Control}, pages 2895--2900,
  2014.

\bibitem{yxchen}
Y.~Chen, S.~Z. Khong, and T.~T. Georgiou.
\newblock On the definiteness of graph {L}aplacians with negative weights:
  Geometrical and passivity-based approaches.
\newblock In {\em Proc. 2016 Amer. Control Conf.}, pages 2488--2493, 2016.

\bibitem{zelazo2017}
D.~Zelazo and M.~B{\"u}rger.
\newblock On the robustness of uncertain consensus networks.
\newblock {\em IEEE Trans. Control Netw. Syst.}, 4(2):170--178, 2017.

\bibitem{pan2016laplacian}
L.~Pan, H.~Shao, and M.~Mesbahi.
\newblock {L}aplacian dynamics on signed networks.
\newblock In {\em Proc. 55th IEEE Conf. Decision Control}, pages 891--896,
  2016.

\bibitem{ahmadizadeh2017eigenvalues}
S.~Ahmadizadeh, I.~Shames, S.~Martin, and D.~Ne{\v{s}}i{\'c}.
\newblock On eigenvalues of {L}aplacian matrix for a class of directed signed
  graphs.
\newblock {\em Linear Algebra Appl.}, 523:281--306, 2017.

\bibitem{zhang2017bipartite}
H.~Zhang and J.~Chen.
\newblock Bipartite consensus of multi-agent systems over signed graphs: State
  feedback and output feedback control approaches.
\newblock {\em Int. J. Robust Nonlinear Control}, 27(1):3--14, 2017.

\bibitem{kunegis2010spectral}
J.~Kunegis, S.~Schmidt, A.~Lommatzsch, J.~Lerner, E.~W.~D Luca, and
  S.~Albayrak.
\newblock Spectral analysis of signed graphs for clustering, prediction and
  visualization.
\newblock In {\em Proc. 2010 SIAM International Conference on Data Mining},
  pages 559--570, 2010.

\bibitem{nishikawa2010network}
T.~Nishikawa and A.~E. Motter.
\newblock Network synchronization landscape reveals compensatory structures,
  quantization, and the positive effect of negative interactions.
\newblock {\em Proc. Natl. Acad. Sci.}, 107(23):10342--10347, 2010.

\bibitem{motter2013spontaneous}
A.~E. Motter, S.~A. Myers, M.~Anghel, and T.~Nishikawa.
\newblock Spontaneous synchrony in power-grid networks.
\newblock {\em Nat. Phys.}, 9(3):191--197, 2013.

\bibitem{song2017network}
Y.~Song, D.~J. Hill, and T.~Liu.
\newblock Network-based analysis of small-disturbance angle stability of power
  systems.
\newblock {\em IEEE Trans. Control Netw. Syst.}, 5(3):901--912, 2018.

\bibitem{ding2017impact}
T.~Ding, R.~Bo, Y.~Yang, and F.~Blaabjerg.
\newblock Impact of negative reactance on definiteness of {B}-matrix and
  feasibility of {DC} power flow.
\newblock {\em IEEE Trans. Smart Grid}, 10(2):1725--1734, 2017.

\bibitem{lien2000dual}
M.~Lien and W.~Watkins.
\newblock Dual graphs and knot invariants.
\newblock {\em Linear Algebra Appl.}, 306(1-3):123--130, 2000.

\bibitem{la1}
J.~C. Bronski and L.~Deville.
\newblock Spectral theory for dynamics on graphs containing attractive and
  repulsive interactions.
\newblock {\em SIAM J. Appl. Math.}, 74(1):83--105, 2014.

\bibitem{HorJoh85}
R.~A. Horn and C.~R. Johnson.
\newblock {\em Matrix Analysis}.
\newblock Cambridge University Press, Cambridge, 1985.

\bibitem{noutsos2008reachability}
D.~Noutsos and M.~J. Tsatsomeros.
\newblock Reachability and holdability of nonnegative states.
\newblock {\em SIAM J. Matrix Anal. \& Appl.}, 30(2):700--712, 2008.

\bibitem{noutsoslaa}
D.~Noutsos.
\newblock On {P}erron-{F}robenius property of matrices having some negative
  entries.
\newblock {\em Linear Algebra Appl.}, 412:132--153, 2006.

\bibitem{elhashash2008generalizations}
A.~Elhashash and D.~B. Szyld.
\newblock Generalizations of {M}-matrices which may not have a nonnegative
  inverse.
\newblock {\em Linear Algebra Appl.}, 429(10):2435--2450, 2008.

\bibitem{olesky2009m}
D.~D. Olesky, M.~J. Tsatsomeros, and P.~van~den Driessche.
\newblock $\text{M}_{\vee}$-matrices: {A} generalization of {M}-matrices based
  on eventually nonnegative matrices.
\newblock {\em Electron. J. Linear Algebra}, 18(1):339--351, 2009.

\bibitem{jiang2016sign}
Y.~Jiang, H.~Zhang, and J.~Chen.
\newblock Sign-consensus of linear multi-agent systems over signed directed
  graphs.
\newblock {\em IEEE Trans. Ind. Electron.}, 64(6):5075--5083, 2016.

\bibitem{ChenCDC16}
W.~Chen, J.~Liu, Y.~Chen, S.~Z. Khong, D.~Wang, T.~Ba\c{s}ar, L.~Qiu, and K.~H.
  Johansson.
\newblock Characterizing the positive semidefiniteness of signed {L}aplacians
  via effective resistances.
\newblock In {\em Proc. 55th IEEE Conf. Decision Control}, pages 985--990,
  2016.

\bibitem{chen2016semidefiniteness}
W.~Chen, D.~Wang, J.~Liu, T.~Ba{\c{s}}ar, K.~H. Johansson, and L.~Qiu.
\newblock On semidefiniteness of signed {L}aplacians with application to
  microgrids.
\newblock In {\em Proc. 6th IFAC Workshop on Distributed Estimation and Control
  in Networked Systems}, pages 97--102, 2016.

\bibitem{chen2017spectral}
W.~Chen, D.~Wang, J.~Liu, T.~Ba{\c{s}}ar, and L.~Qiu.
\newblock On spectral properties of signed {L}aplacians for undirected graphs.
\newblock In {\em Proc. 56th IEEE Conf. Decision Control}, pages 1999--2002,
  2017.

\bibitem{berman1994nonnegative}
A.~Berman and R.~J. Plemmons.
\newblock {\em Nonnegative Matrices in the Mathematical Sciences}.
\newblock SIAM, 1994.

\bibitem{chiang1989closest}
H.-D. Chiang and J.~S. Thorp.
\newblock The closest unstable equilibrium point method for power system
  dynamic security assessment.
\newblock {\em IEEE Trans. Circuits Systems}, 36(9):1187--1200, 1989.

\bibitem{chiang1987foundations}
H.-D. Chiang, F.~Wu, and P.~Varaiya.
\newblock Foundations of direct methods for power system transient stability
  analysis.
\newblock {\em IEEE Trans. Circuits Systems}, 34(2):160--173, 1987.

\bibitem{Murray2004}
R.~Olfati-Saber and R.~M. Murray.
\newblock Consensus problems in networks of agents with switching topology and
  time-delays.
\newblock {\em IEEE Trans. Autom. Control}, 49(9):1520--1533, 2004.

\bibitem{xia2016}
W.~Xia, M.~Cao, and K.~H. Johansson.
\newblock Structural balance and opinion separation in trust-mistrust social
  networks.
\newblock {\em IEEE Trans. Control Netw. Syst.}, 3(1):46--56, 2016.

\bibitem{Florian2013}
F.~D\"{o}rfler and F.~Bullo.
\newblock Kron reduction of graphs with applications to electrical networks.
\newblock {\em IEEE Trans. Circuits Systems-I}, 60:150--163, 2013.

\bibitem{laws}
C.~R. Paul.
\newblock {\em Fundamentals of Electric Circuit Analysis}.
\newblock John Wiley \& Sons, 2001.

\bibitem{curtis1998circular}
E.~B. Curtis, D.~Ingerman, and J.~A. Morrow.
\newblock Circular planar graphs and resistor networks.
\newblock {\em Linear Algebra Appl.}, 283(1-3):115--150, 1998.

\bibitem{KleRan93}
D.~J. Klein and M.~Randi\'{c}.
\newblock Resistance distance.
\newblock {\em J. Math. Chem.}, 12(1):81--95, 1993.

\bibitem{Aybat2017DecentralizedCO}
N.~S. Aybat and M.~G{\"u}rb{\"u}zbalaban.
\newblock Decentralized computation of effective resistances and acceleration
  of consensus algorithms.
\newblock In {\em 2017 IEEE Global Conference on Signal and Information
  Processing}, pages 538--542, 2017.

\bibitem{BaoLee07}
J.~Bao and P.~L. Lee.
\newblock {\em Process Control: The Passive Systems Approach}.
\newblock Advances in Industrial Control. Springer, 2007.

\bibitem{Anderson73}
B.~D.~O. Anderson and S.~Vongpanitlerd.
\newblock {\em Network Analysis and Synthesis: A Modern Systems Theory
  Approach}.
\newblock Prentice Hall, 1973.

\bibitem{zhang2006schur}
F.~Zhang, editor.
\newblock {\em The Schur Complement and Its Applications}.
\newblock Springer Science \& Business Media, 2006.

\bibitem{ASMNc}
S.~Ahmadizadeh, I.~Shames, S.~Martin, and D.~Ne{\v{s}}i{\'c}.
\newblock Corrigendum to ``{O}n eigenvalues of {L}aplacian matrix for a class
  of directed signed graphs'' [{L}inear {A}lgebra {A}ppl. 523 (2017) 281-306].
\newblock {\em Linear Algebra Appl.}, 530:541--557, 2017.

\bibitem{mukherjee2016consensus}
D.~Mukherjee and D.~Zelazo.
\newblock Consensus over weighted digraphs: {A} robustness perspective.
\newblock In {\em Proc. 55th IEEE Conf. Decision Control}, pages 3438--3443,
  2016.

\bibitem{YSL16p1}
G.~Young, L.~Scardovi, and N.~Leonard.
\newblock A new notion of effective resistance for directed graphs - {P}art
  {I}: Definition and properties.
\newblock {\em IEEE Trans. Autom. Control}, 61(7):1727--1736, 2016.

\bibitem{YSL16p2}
G.~Young, L.~Scardovi, and N.~Leonard.
\newblock A new notion of effective resistance for directed graphs - {P}art
  {II}: Computing resistances.
\newblock {\em IEEE Trans. Autom. Control}, 61(7):1737--1752, 2016.

\end{thebibliography}
\end{document}